\documentclass{article}
\usepackage[utf8]{inputenc}
\usepackage{amsmath}
\usepackage{dsfont}
\usepackage{comment}
\usepackage{amssymb}
\usepackage{amsthm}
\usepackage{mathtools}
\usepackage {tikz}
\usepackage[margin=1in]{geometry}
\usepackage{systeme}
\usepackage{colortbl}
\usepackage{nicefrac}
\usepackage{multicol}
\usepackage{multirow}
\usepackage{xspace}
\usepackage{enumitem}
\usepackage{hyperref}
\usepackage{mdframed}
\usepackage{makecell}

\hypersetup{colorlinks=true,linkcolor=black,citecolor=[rgb]{0.5,0,0}}
\newtheorem{theorem}{Theorem}[section]

\newtheorem{corollary}[theorem]{Corollary}
\newtheorem{lemma}[theorem]{Lemma}
\newtheorem{claim}{Claim}[section]
\theoremstyle{definition}
\newtheorem{definition}{Definition}[section]
\theoremstyle{definition}

\newcommand{\tO}{\widetilde{O}}
\newcommand{\Oish}{\widetilde{O}}

\newcommand{\sss}{\textsc{SSS}}

\newcommand{\yes}{\texttt{YES}}
\newcommand{\no}{\texttt{NO}}

\newcommand{\dist}{\texttt{dist}}
\newcommand{\eps}{\varepsilon}

\usepackage{todonotes}

\usepackage{setspace}
\usepackage{authblk}

\title{Spanning Adjacency Oracles in Sublinear Time
\thanks{This work was supported by NSF:AF 2153680.}
}
\author[1]{Greg Bodwin}
\author[2]{Henry Fleischmann}
\date{}
\affil[1]{University of Michigan EECS. \texttt{bodwin@umich.edu}}
\affil[2]{University of Cambridge DPMMS. \texttt{hfleischmann3@gmail.com}}
\date{}

\begin{document}
\pagenumbering{gobble}

\maketitle

\begin{abstract}
Suppose we are given an $n$-node, $m$-edge input graph $G$, and the goal is to compute a spanning subgraph $H$ on $O(n)$ edges.
This can be achieved in linear $O(m + n)$ time via breadth-first search.
But can we hope for \emph{sublinear} runtime in some range of parameters---for example, perhaps $O(n^{1.9})$ worst-case runtime, even when the input graph has $n^2$ edges?

If the goal is to return $H$ as an adjacency list, there are simple lower bounds showing that $\Omega(m + n)$ runtime is necessary.
If the goal is to return $H$ as an adjacency matrix, then we need $\Omega(n^2)$ time just to write down the entries of the output matrix.
However, we show that neither of these lower bounds still apply if instead the goal is to return $H$ as an \emph{implicit} adjacency matrix, which we call an \emph{adjacency oracle}.
An adjacency oracle is a data structure that gives a user the illusion that an adjacency matrix has been computed: it accepts edge queries $(u, v)$, and it returns in near-constant time a bit indicating whether or not $(u, v) \in E(H)$.

Our main result is that, for any $0 < \eps < 1$, one can construct an adjacency oracle for a spanning subgraph on at most $(1+\eps)n$ edges, in $\Oish(n \eps^{-1})$ time (hence sublinear time on input graphs with $m \gg n$ edges), and that this construction time is near-optimal.
Additional results include constructions of adjacency oracles for $k$-connectivity certificates and spanners, which are similarly sublinear on dense-enough input graphs.

Our adjacency oracles are closely related to Local Computation Algorithms (LCAs) for graph sparsifiers; they can be viewed as LCAs with some computation moved to a preprocessing step, in order to speed up queries.
Our oracles imply the first LCAs for computing sparse spanning subgraphs of general input graphs in $\Oish(n)$ query time, which works by constructing our adjacency oracle, querying it once, and then throwing the rest of the oracle away.
This addresses an open problem of Rubinfeld [CSR '17].
\end{abstract}

\clearpage
\tableofcontents

\clearpage
\pagenumbering{arabic}
\section{Introduction}


A \emph{sparsifier} of a graph $G$ is a smaller graph $H$ that approximately preserves some important structural properties of $G$.
Examples include spectral sparsifiers, flow/cut sparsifiers, spanners, preservers, etc.
Let us focus for now on the computation of a particularly simple kind of sparsifier, which we call a \emph{sparse spanning subgraph}:
\begin{mdframed}[backgroundcolor=gray!20]
\texttt{Sparse Spanning Subgraph} (\sss{}):
\begin{itemize}
\item \textbf{Input:} An $n$-node, $m$-edge, undirected graph $G = (V, E)$ and $\eps > 0$.
\item \textbf{Output:} An edge-subgraph $H$ on at most $(1+\eps) n$ edges that spans $G$.
\end{itemize}
\end{mdframed}

An \sss{} is a slightly relaxed version of a spanning forest,\footnote{An edge-subgraph $H$ \emph{spans} a graph $G$ if it has the same connected components as $G$.  A \textit{spanning forest} of $G$ is any forest that spans $G$, i.e., it is the union of spanning trees for each connected component of $G$.} and so this problem can be solved in linear $O(m + n)$ time via breadth-first search (BFS).
We might wonder for a moment whether linear runtime for this basic algorithm is \emph{optimal}.
Especially for denser input graphs, we might dream of \emph{sublinear-time} algorithms, which try to discover one of the many valid spanning subgraphs without even reading most of the input graph.
Can we solve \sss{} in, say, $O(n^{1.9})$ worst-case time, even when the input graph has $\Theta(n^2)$ edges?

Unfortunately, the canonical answer is \emph{no}: we cannot hope to solve \sss{} without reading at least a constant fraction of the input graph.
The counterexamples generally work by planting a cut edge into an otherwise-random graph.
For example, in a lower bound construction from \cite{LRR20, PRVY19}, we construct $G$ from two node-disjoint random graphs $G_1, G_2$ on $n/2$ nodes each, which include each possible edge independently with probability $1/2$, plus a single ``cut edge'' $e$ connecting a random node from $G_1$ to a random node from $G_2$.
We absolutely must take $e$ in the spanning subgraph $H$, but we need to scan essentially the entire input graph just to find $e$.

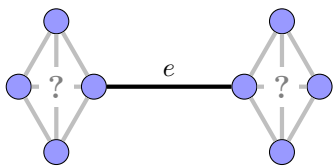
\begin{figure} [h]
\begin{center}
\begin{tikzpicture}[scale=1]
    \node (a) at (0,0) [circle,draw=black,fill=blue!40] {};
    \node (b) at (1,0) [circle,draw=black,fill=blue!40] {};
    \node (c) at (0.5,0.87) [circle,draw=black,fill=blue!40] {};
    \node (d) at (0.5,-0.87) [circle,draw=black,fill=blue!40] {};
    
    \draw [ultra thick, gray!50!white] (a) -- (b) -- (c) -- (d) -- (a) -- cycle;
    \draw [ultra thick, gray!50!white] (d) -- (b);
    \draw [ultra thick, gray!50!white] (c) -- (a);
    \node [gray, fill=white] at (0.5, 0) {\bf \large ?};

    \node (e) at (3,0) [circle,draw=black,fill=blue!40] {};
    \node (f) at (4,0) [circle,draw=black,fill=blue!40] {};
    \node (g) at (3.5,0.87) [circle,draw=black,fill=blue!40] {};
    \node (h) at (3.5,-0.87) [circle,draw=black,fill=blue!40] {};
    
    \draw [ultra thick, gray!50!white] (e) -- (f) -- (g) -- (h) -- (e);
    \draw [ultra thick, gray!50!white] (e) -- (g);
    \draw [ultra thick, gray!50!white] (f) -- (h);
    \node [gray, fill=white] at (3.5, 0) {\bf \large ?};

    \draw[ultra thick] (b) -- (e) node [midway, above] {$e$};
\end{tikzpicture}
\end{center}
\end{figure}

The starting point of this paper is that the cut-edge lower bound, while formidable, is actually a bit restricted in scope: it does not quite apply in all graph representation models.
We explain this next.

\subsection{Adjacency Oracles and Algorithms for \sss{}}

There are two popular ways to represent graphs:
\begin{enumerate}
\item First, the \textbf{adjacency list representation} is an array of arrays.
If we index into an adjacency list $L$ with a node $v$, then $L[v]$ returns a list of the neighbors of $v$ in some arbitrary order.\footnote{We assume for this discussion that the list is given in a form where we can check its length (corresponding to $\deg(v)$) and where we can query a random neighbor, both in $\Oish(1)$ time.}

\item Second, the \textbf{adjacency matrix representation} is an $n \times n$ matrix, which holds a $1$ or $0$ in each position $(u, v)$ to represent whether $(u, v)$ is or is not an edge in the graph.
\end{enumerate}

The cut-edge lower bound for \sss{} shows that we cannot hope to return a spanning subgraph $H$ \emph{as an adjacency list} in sublinear time.
For adjacency matrices, the situation is even worse: it takes $\Omega(n^2)$ time just to fill out the entries of the matrix, so no sublinear algorithms are possible.
However, this paper will consider a tweak on the adjacency matrix model:

\begin{definition} [Adjacency Oracles]
An implicit adjacency matrix for a graph $G$, which we will call an \textbf{adjacency oracle}, is a data structure that, on query $(u, v)$, deterministically returns a bit in $\Oish(1)$ time\footnote{Here and throughout the paper, $\Oish(\cdot)$ notation hides $\text{polylog}(n)$ factors.} indicating whether or not $(u, v) \in E(G)$.
\end{definition}

The determinism in queries is essential to force the oracle to be \emph{history-independent}; that is, the graph $G$ that the oracle represents must be fixed at the end of the construction phase, and it may not depend on the order in which the oracle receives queries.
The point of the adjacency oracle model is that it is not constrained by the $\Omega(n^2)$ output size lower bound---the data structure could, in principle, have $\ll n^2$ bits---and for more subtle reasons, it escapes the cut-edge lower bound as well.
The hard part of the cut-edge lower bound is scanning to find the cut edge $e$, but when we build an adjacency oracle it is the \emph{user} who does the hard work of pointing out the edge $e$.
In other words, an adjacency oracle needs to quickly \emph{recognize} that the cut-edge $e$ is necessary for the spanning subgraph when it is received at query time, but this is potentially an easier task than proactively \emph{discovering} $e$ at preprocessing time.

Given this, we can now reopen the question of whether BFS is the most efficient algorithm to solve \sss{}, when we allow the output subgraph to be represented by an adjacency oracle.
Our first main result is that, in fact, it is not.

\begin{theorem} [Adjacency Oracles for \sss{}] \label{thm:introsss} ~
\begin{itemize}
\item \textbf{(Upper Bound)} For any $0 < \eps < 1$ and $n$-node input graph, there is a randomized algorithm that solves \sss{} with high probability in $\Oish(n\varepsilon^{-1})$ time, where the output subgraph $H$ is returned as an adjacency oracle.

\item \textbf{(Lower Bound)} However, no algorithm as above can run in $O(n^{1-\delta})$ time, for any constant $\delta>0$.
\end{itemize}
\end{theorem}

As a point of clarification, we assume that the input graph $G$ is received in all useful forms, i.e., as both an adjacency list and an adjacency matrix/oracle.
The adjacency oracle for $H$ is allowed access to the adjacency oracle for $G$, so that if queried with an edge $(u, v) \notin E(G)$ it can quickly say \texttt{NO}.
However, we note that this is the \emph{only} way in which our adjacency oracles use access to $G$, and so they still works in a slightly stronger model where the adjacency oracle for $H$ may not access $G$ in any way but the user promises a priori to only query the oracle with edges from $E(G)$.


This model, in which we obtain sublinear algorithms by providing near-constant-time oracle query access to the output (and where these queries may access the input), is the typical one in sublinear algorithms.
It is sometimes called a \emph{solution oracle}, and it has been used previously as a paradigm for sublinear algorithms for vertex cover, dominating set, maximum matching, independent set, and others \cite{VY21, HKNO09, NO08, YYI12}.
It is also analogous to the model used for Local sparsifier algorithms, which we will discuss in detail shortly.
Indeed, the lower bound part of Theorem \ref{thm:introsss} uses a graph construction and analysis developed in the context of Local sparsifier algorithms for most of its heavy lifting \cite{PRVY19, LRR20}.

\subsection{Adjacency Oracles for $k$-Connectivity Certificates}

A natural generalization of a sparse spanning subgraph is a \emph{$k$-connectivity certificate}:
\begin{definition} [$k$-Connectivity Certificates]
Given a graph $G$, a subgraph $H$ is a $k$-connectivity certificate if, for any edge set $F \subseteq E(G), |F| \le k-1$, the connected components of $G \setminus F$ and $H \setminus F$ are identical.
\end{definition}

The following are two equivalent definitions of $k$-connectivity certificates often used in the literature.
We will use these instead of the above definition where convenient (e.g., Lemma \ref{lem:kcccorrect2}).
\begin{itemize}
\item For any integer $r$ and any cut $C$ in $G$ of size $r$, the size of $C$ in $H$ is at least $\min\{k, r\}$.
\item For any integer $r$ and nodes $s, t$ that are $r$-connected in $G$, they are at least $\min\{k, r\}$-connected in $H$.
\end{itemize}
A $1$-connectivity certificate is the same as a spanning subgraph, and in this sense the problem of constructing sparse $k$-connectivity certificates generalizes \sss{}. For this reason, we refer to the problem as $k$-\sss{}, with $1$-\sss{} and \sss{} identical.
The size bounds for $k$-connectivity certificates also extend those for spanning forests:
\begin{theorem} [\cite{NagamochiI92}]
Every $n$-node graph $G$ has a $k$-connectivity certificate $H$ of size $|E(H)| \le k (n-1)$.
\end{theorem}

We show:
\begin{theorem} [Adjacency Oracles for $k$-Connectivity Certificates]
For any $\eps > 0$, $n$-node input graph $G$, and $k = \Oish(1)$, there is a randomized algorithm that with high probability computes an adjacency oracle for a $k$-connectivity certificate $H \subseteq G$ of size $|E(H)| \le (1+\eps)kn$, and which runs in $\Oish(n \eps^{-1})$ time.
\end{theorem}

\subsection{Adjacency Oracles vs.\ Local Sparsifier Algorithms}

We will temporarily pause discussion of our results to discuss their relationship to Local Computation Algorithms (LCAs) for graph sparsifiers, which are the conceptually closest prior work to ours.
LCAs were introduced in a classic paper by Rubinfeld, Tamir, Vardi, and Xie \cite{RTVX11}.
The high-level goal is to design an algorithm which, on input $(X, i)$, computes the $i^{th}$ part of the output associated to input $X$, ideally in sublinear time.
Interesting LCAs have been discovered for many problems, such as graph coloring, SAT solving, graph sparsifiers, etc.~\cite{ARVX12, Solomon21, VY21, EMR18, HMV16, LRY15, MPV18, RV16, MV13, MRVX12}; for more, see the survey of Levi and Medina \cite{LM17}.

An LCA algorithm for \sss{} would have input $(G, \eps, e)$, where $e \in E(G)$, and it would return either \yes{} or \no{} in such a way that for any fixed $(G, \eps)$ the set of edges $e$ to which the algorithm says \yes{} form a valid solution to \sss{}.
As in our setting, the input graph $G$ is typically given to an LCA as both an adjacency list and an adjacency oracle.
The main technical difference between adjacency oracles and LCAs for \sss{} is essentially whether the focus is placed on preprocessing or query time.
That is:
\begin{itemize}
\item LCAs do not allow any centralized preprocessing before answering queries $(G, \eps, e)$ (with the minor exception that they typically allow shared randomness across queries).
On the other hand, adjacency oracles for \sss{} allow centralized preprocessing, and the goal of the problem is to minimize this preprocessing time.

\item Adjacency Oracles insist on $\Oish(1)$ query time in order to provide the illusion that we are indexing into an adjacency matrix.
On the other hand, for LCAs the query time may be much larger, and the goal of the problem is to minimize this query time.
\end{itemize}

Thus, adjacency oracles might be viewed as an investigation of the extent to which a preprocessing phase can improve the query time for LCAs.
These differences are enough to formally separate the two models.
Levi, Ron, and Rubinfeld \cite{LRR20} proved that any LCA for \sss{} requires $\Omega(n^{1/2})$ query time.
Together with Theorem \ref{thm:introsss}, this implies a strong separation, i.e., LCAs cannot achieve $\Oish(1)$ query time as in the adjacency oracles of Theorem \ref{thm:introsss} (which bypass the lower bound of \cite{LRR20} due to their use of centralized preprocessing).
An even stronger form of this lower bound was later proved by Parter, Rubinfeld, Vakilian, and Yodpinyanee \cite{PRVY19}.
There is also a natural way to interpolate between the adjacency oracle and LCA models---allowing super-constant preprocessing and query time and investigating the tradeoff between them---but we will leave this as a possible direction for future work.

The optimal runtime of LCAs for \sss{} is a fascinating open problem; we refer to the excellent survey of Rubinfeld \cite{Rubinfeld17} for an in-depth presentation.
In particular, this survey highlights the following two open questions:

\paragraph{Open Question (c.f.~\cite{Rubinfeld17}, Problem 1).} It is known that every graph with constant maximum degree or high expansion has an LCA for \sss{} in $O(n)$ query time \cite{LRR20, LMRRS17, LL18, LRR16}.
Is there an LCA algorithm for \sss{} that works on any sparse input graph ($O(n)$ edges), and which runs in $o(n)$ query time?

\paragraph{Open Question (c.f.~\cite{Rubinfeld17}, Section 4).} Rubinfeld writes that \emph{``nothing is known [about LCAs for \sss{}] when there is no bound on the maximum degree of the input graph.''}
Subsequent work \cite{ACLP23, PRVY19, LL18} discovered algorithms for general graphs of maximum degree $\Delta$ with query complexity $\Oish(n^{2/3} \cdot \text{poly} \Delta )$, but this is still nontrivial only when $\Delta$ is a small-enough polynomial in $n$.

\vspace{0.5cm}

\noindent Our new adjacency oracles partially address these questions:
\begin{corollary} \label{cor:introlocal}
There is an LCA for \sss{} that runs in $\Oish(n \eps^{-1})$ query time.
(This algorithm works for any input graph.)
\end{corollary}
\begin{proof}
On query $(u, v)$, use Theorem \ref{thm:introsss} to construct an adjacency oracle for \sss{} in $\Oish(n \eps^{-1})$ time (using a shared tape of random bits across all possible queries).
Then query the adjacency oracle on $(u, v)$ and return the result.
\end{proof}

Corollary \ref{cor:introlocal} positively addresses the latter open question.
Although it does not resolve the former open question, it may shed light on a way forward: (1) our algorithm does not use the assumption that the input graph is sparse, and (2) our algorithm performs \emph{the same} $\Oish(n)$ work for each query, to compute the adjacency oracle, and then only the last $\Oish(1)$ steps where the adjacency oracle is accessed differ between queries.
It seems unlikely that the optimal LCA for \sss{} would be so query-independent, and so perhaps improved LCAs could be achieved by exploiting knowledge of the query $(u, v)$ at an earlier stage of the construction.

Our adjacency oracle upper bounds for \sss{} and $k$-\sss{} are not technically similar to prior work on LCAs.
However, the rest of this work connects to the LCA literature at a technical level as well, and draws on several clever constructions and techniques developed in the context of LCAs.
In particular, our lower bound for \sss{} adjacency oracles uses the lower bound construction from \cite{LRR20} as a black box, and our results on adjacency oracles for spanners (which we discuss next) draw on ideas used in LCAs for graph spanners \cite{ACLP23}.





\subsection{Adjacency Oracles for Graph Spanners}

Besides sparse spanning subgraphs, one can more strongly ask for \emph{spanners}, which preserve approximate distances rather than just connectivity.
\begin{definition} [Spanners \cite{PU89jacm, PU89sicomp}]
Given a graph $G$, a $k$-spanner is an edge-subgraph $H$ satisfying $\dist_H(s, t) \le k \cdot \dist_G(s, t)$ for all nodes $s, t \in V$.
\end{definition}

We show the following two results for computing spanners:

\begin{theorem}[Adjacency Oracles for $3$-Spanners]
For any $n$-node input graph, there is a randomized algorithm that with high probability computes an adjacency oracle for a $3$-spanner $H$ of size $|E(H)| = \Oish(n^{3/2})$, and which runs in $\Oish(n^{3/2})$ time.
\end{theorem}

\begin{theorem}[Adjacency Oracles for $5$-Spanners]
For any $n$-node input graph, there is a randomized algorithm that with high probability computes an adjacency oracle for a $5$-spanner $H$ of size $|E(H)| = \Oish(n^{4/3})$, and which runs in $\Oish(n^{3/2})$ time.
\end{theorem}

Both of these results are proved only for unweighted input graphs.
The sizes of these $3$- and $5$-spanners are optimal, up to hidden log factors, and can be viewed as sublinear computation for input graphs on $\gg n^{3/2}$ edges. 

The reason that we have results for $3$- and $5$-spanners, but not spanners of higher stretch, is due to a common technical barrier.
These results follow a construction strategy of Baswana and Sen \cite{BS07}, based on hierarchical clustering; $3$- and $5$-spanners can be achieved using only one level of clustering, while higher-stretch spanners of optimal size require two or more levels of clustering.
We are able to optimize the first cluster assignment step, but it is unclear how to achieve sublinear time cluster assignment for depth $2$ and beyond. 

Nonetheless, by introducing edge-sampling to the algorithm of Baswana and Sen, we show a construction for higher-stretch spanners, with a parameter $\rho$ controlling the tradeoff between  size and preprocessing time. Our algorithm runs in a factor of $1/\rho$ less than linear time $\Oish(m)$ time by incurring a factor of $\rho^2$ edges over the optimal size (as usual, this optimality is conditional on the girth conjecture \cite{girth}).

\begin{theorem}[Adjacency Oracles for $(2k-1)$-spanners]
For any $n$-node, $m$-edge input graph and $k = \Oish(1)$,  there exists a randomized algorithm that with high probability computes an adjacency oracle for a $(2k-1)$-spanner $H$ of size $|E(H)| = \Oish(n^{1 + 1/k}\rho^2)$, and which runs in $\Oish(n + m/\rho)$ time.
\end{theorem}

Our construction leaves open the question of whether adjacency oracles of optimal size spanners can be computed in sublinear time.

\paragraph{Open Question.} For $k > 3$, is there a randomized algorithm that computes the adjacency oracle of an $\Oish(n^{1 + 1/k})$ size $(2k-1)$-spanner in polynomially subquadratic time?\\

This roughly mirrors a difficulty faced in the corresponding local spanner algorithms, where optimal results were recently achieved for $3$- and $5$-spanners in a nice paper by Arviv, Chung, Levi, and Pyne \cite{ACLP23}, but where results for higher-stretch spanners are considerably more restricted in scope.
Indeed, \cite{ACLP23} is also based on Baswana-Sen clustering, and it introduces a helpful degree-bucketing technique that we adapt and refine for our setting.

\section{Spanning Adjacency Oracles}

\subsection{Adjacency Oracles for Sparse Spanning Subgraphs} \label{sec: near-linear time SSS}


We now construct our adjacency oracle for the problem \sss{}, defined in the introduction.



\paragraph{Preprocessing Algorithm.} We first describe the data structures that we create in the construction of the adjacency oracle.
We will use three data structures:
\begin{itemize}
\item Any kind of set data structure, to represent the edges added to the subgraph $H$ during preprocessing.
This data structure only needs to support insertions and queries (checking whether or not an edge $(u, v)$ has been previously inserted to $E(H)$) in $\Oish(1)$ time each.
Many data structures are available that achieve this behavior, e.g., a self-balancing binary search tree suffices.

\item We use a union-find data structure to maintain connected components $\{C_i\}$.
Initially, each node is in its own connected component, but we will iteratively merge components throughout the algorithm.
This requries $\Oish(n)$ total time.

\item Another data structure is used to perform a particular edge-sampling step in preprocessing; it will be easier to describe this in Lemma \ref{lem:edgesample} following our description of the algorithm itself.
\end{itemize}

We will say that a node $v$ is ``in bucket $b$" at a moment in the algorithm if the connected component $C_v$ containing $v$ currently has size $2^b \le |C_v| < 2^{b+1}$.
So, initially every node is in bucket $0$, and nodes may be promoted to higher buckets when their components are merged with other components.

\begin{mdframed}[backgroundcolor=gray!20]
\textbf{\sss{} Adjacency Oracle Preprocessing}

\begin{itemize}
\item Let $b \gets 0$.
This is an incremental counter that marks the current bucket we are processing.

\item While $b < \log n$:
\begin{itemize}
\item Let $E_b \subseteq E(G)$ be the set of edges for which at least one of the two endpoints is in bucket $b$.
Choose an edge $(u, v) \in E_b$ uniformly at random.
(Taking a uniform-random sample from $E_b$ in $\Oish(1)$ time is nontrivial; we address this in Lemma \ref{lem:edgesample} below.)

\item Let $u$ be the endpoint of the sampled edge $(u, v)$ that is in bucket $b$.
If $v$ is currently in a different connected component than $u$, and also $v$ is in a bucket $b' \ge b$, then we call this sample a \emph{success}.
Otherwise, if $u, v$ are in the same component or if $b' < b$, then we call this sample a \emph{failure}.
\begin{itemize}
\item If the sample is a failure, do nothing.

\item If the sample is a success, add the edge $(u, v)$ to $H$, and merge the connected components containing $u$ and $v$.
Note that this increases the bucket of $v$ (and possibly also the bucket of $u$), and so it changes the set of nodes in bucket $b$, and therefore also the set of edges in $E_b$.
\end{itemize}

\item Repeat until we have $c \eps^{-1} 2^b \log^2 n$ failure events in a row (where $c>0$ is a sufficiently large absolute constant that we leave implicit).
Then set $b \gets b+1$, i.e., we move on to analyzing the next bucket.
\end{itemize}
\end{itemize}
\end{mdframed}

We will address the step of sampling uniformly from $E_b$, drawing from similar work on dynamic weighted sampling \cite{HMM93}.

\begin{lemma} \label{lem:edgesample}
For each bucket index $b$, we can maintain a data structure that allows us to sample uniformly from the edges currently in $E_b$ in $\Oish(1)$ time with high probability,\footnote{Here and throughout the paper, ``high probability'' means probability $\ge 1 - 1/n^c$ for an absolute constant $c>0$.} using $\Oish(n)$ total update time.
\end{lemma}
\begin{proof}
First we describe the creation and maintenance of the data structure.
When we begin analyzing bucket index $b$, we first scan all nodes and make a list of the nodes in bucket $b$ as well as their degrees.
We sort these nodes into $\log n$ \emph{groups} $\{\Gamma_1, \dots, \Gamma_{\log n}\}$ by their degrees: group $\Gamma_i$ contains the nodes in bucket $b$ whose degree falls in the range $[2^i, 2^{i+1})$.
We will also maintain the sum of node degrees in each group, which we will write as $\deg(\Gamma_i)$.
It takes $\Oish(n)$ time to create these groups at the beginning of our analysis of bucket $b$, and as nodes leave bucket $b$ due to connected component merges, we can straightforwardly remove them from the corresponding group and update our size and degree counts in $\Oish(1)$ time per node.

Now we describe the sampling algorithm.
We execute the following process:
\begin{itemize}
\item Choose a group $\Gamma_i$ with probability proportional to $\deg(\Gamma_i)$, i.e., $\deg(\Gamma_i)/ \sum_j \deg(\Gamma_j)$.

\item Choose a node $u \in \Gamma_i$ uniformly at random.
With probability $\deg(u) / 2^{i+1}$ we accept $u$ and move on to the next step; otherwise, repeat this step, selecting a new node $u \in \Gamma_i$ uniformly at random.

\item Choose an edge $(u, v)$ incident to the node $u$ selected in the previous step, uniformly at random.
Then, check whether the other endpoint $v$ is in bucket $b$.
If not, return $(u, v)$ as the sampled edge.
If so, then with probability $1/2$ return $(u, v)$ as the sampled edge, and with probability $1/2$ go back to the beginning of the entire sampling process and repeat from scratch.
\end{itemize}

Note that we accept our node sample in the second step with probability at least $1/2$, and so with high probability we repeat this step at most $\Oish(1)$ times.
Similarly, in the third step we return an edge with probability at least $1/2$, so with high probability we restart the entire sampling process only $\Oish(1)$ times.
Together, this implies that each sample runs in $\Oish(1)$ time with high probability.

Finally, we argue that this process selects a uniform-random edge from $E_b$.
Let $E^*_b$ be the set of \emph{oriented} edges $(u, v) \in E(G)$ such that the first node $u$ is in bucket $b$ (if $v$ is also in bucket $b$, then we have both $(u, v), (v, u) \in E^*_b$).
Notice that the first two steps of the sampling procedure, combined with the first part of the third step, select a uniform-random oriented edge from $E^*_b$. This is because we select a node $u$ with probability proportional to $\deg(u)$ (the number of edges in $E^*_b$ that start with $u$), and then we select a uniform-random edge incident to $u$.

This means that an edge $(u, v)$ with both $u, v$ in bucket $b$ is twice as likely to be selected in the third round as an edge with only $u$ in bucket $b$, since such an edge is represented twice (once with each orientation) in $E^*_b$.
The third step resamples these edges with probability $1/2$, so after this step is applied, all edges in $E_b$ are equally likely to be sampled.
\end{proof}

We are now ready to verify the runtime of the preprocessing algorithm:
\begin{lemma}
The above preprocessing algorithm can be implemented to run in $\Oish(n \eps^{-1})$ time (with high probability).
\end{lemma}
\begin{proof}
The union-find data structure has total update time $\Oish(n)$.
We can record the edges added to $H$ in $\Oish(1)$ time each by tracking these edges using any kind of set data structure; we can have at most $n-1$ success events (since each success event causes two components to be merged), and so this takes $\Oish(n)$ time in total.

We next control the number of samples that we take for each of the $\log n$ choices of bucket index $b$.
At the beginning of a round with bucket index $b$, there are at most $n/2^b$ components in bucket $b$, since each component has size at least $2^b$.
Each success event causes at least one component to leave the bucket.
There are at most $\Oish(\eps^{-1} 2^b)$ failure events between success events.
Thus, we sample at most $\Oish(n \eps^{-1})$ times per bucket index.

Finally, by the previous lemma we can sample from $E_b$ in $\Oish(1)$ time with high probability, by paying $\Oish(n)$ additional runtime per bucket index.
\end{proof}

We next look ahead to the step where we bound the number of edges in the subgraph $H$ represented by our adjacency oracle.
The following property of the preprocessing algorithm will be helpful:

\begin{lemma} \label{lem:upedgecount}
With high probability, each time we finish processing a bucket $b$ (i.e.\ just before we set $b \gets b+1$), there are at most $\eps n / \log n$ edges $(u, v) \in E(G)$ with the property that $u$ is in bucket $b$, $v$ is in a different connected component than $u$, and the bucket $b'$ of $v$ satisfies $b' \ge b$.
\end{lemma}
\begin{proof}
Let $Y_b$ be the event that the lemma statement holds for bucket index $b$.
Our first goal is to bound
$$\Pr\left[ Y_b \ \mid \ Y_1 \text{ and } \dots \text{ and } Y_{b-1}\right].$$
To do so, we again let $E^*_b$ be the set of oriented edges $(u, v)$ with $u$ in bucket $b$.
We sort these edges into three types:
\begin{itemize}
\item We say that $(u, v)$ is a \textbf{success edge} if $v$ is in a bucket $b' \ge b$, and $u, v$ are in different connected components.
(That is, sampling a success edge causes a success event, and the goal of this lemma is to bound the number of success edges.)

\item We say that $(u, v)$ is a \textbf{descending edge} if $v$ is in a bucket $b' < b$.

\item We say that $(u, v)$ is an \textbf{internal edge} if $u, v$ are in the same connected component.
\end{itemize}
The number of internal edges can be bounded as
\begin{align*}
\sum \limits_{X \text{ component in bucket } b} |X|^2 \le \ &\frac{n}{2^b} \cdot \left(2^{b+1}\right)^2\\
= \ &O\left(n \cdot 2^b \right).
\end{align*}
To bound the number of descending edges $(u, v)$, note that each such edge would be a success edge for a smaller bucket index if we considered it with reversed orientation $(v, u)$.
Since we have conditioned on $Y_1, \dots, Y_{b-1}$, we assume that the lemma statement holds for all bucket indices $i < b$, we can bound the number of descending edges as
$$\sum \limits_{i < b} \frac{\eps n}{\log n} \le \eps n.$$
So, if there are currently at least $\eps n / \log n$ success edges for bucket $b$, then each time we sample an edge $(u, v)$ the probability that it causes a success event is at least
$$ \Theta\left(\frac{\eps n}{2^b \cdot n \log n}\right) = \Theta\left( \frac{\eps}{2^b \cdot \log n} \right).$$
So by standard Chernoff bounds and by choice of large enough constant $c$, with probability at least $1-1/n^{100}$, if we sample $c \eps^{-1} 2^b \log^2 n$ edges we will have at least one success event.
Since (conservatively) we will have at most $n^2$ edges recorded over the course of the algorithm and hence $O(n^2)$ periods of samples between successes, by a union bound, with probability at least $1-1/n^{98}$ we will not have $c \eps^{-1} 2^b \log^2n$ failure events in a row while $\ge \eps n / \log n$ success edges for bucket $b$ still exist.
Assuming this event occurs, since we only move on to the next bucket $b+1$ following $c \eps^{-1} 2^b \log^2n$ failure events in a row, it must be that $< \eps n / \log n$ success edges remain.

To finish the proof, it now remains only to assemble the previous probability calculations.
We bound
\begin{align*}
\Pr[Y_1 \text{ and } \dots \text{ and } Y_{\log n}]
=&\Pr[Y_1] \cdot \Pr[Y_2 \mid Y_1] \cdot \dots \cdot \Pr[Y_{\log n} \mid Y_1 \text{ and } \dots \text{ and } Y_{\log n - 1}]\\
\ge& \left(1 - 1/n^{98} \right)^{\log n}\\
\ge& 1 - \Theta\left(1/n^{97}\right). \tag*{\qedhere}
\end{align*}
\end{proof}

\paragraph{The Query Algorithm.}

The query algorithm is relatively simple.
As discussed in the introduction, this algorithm is phrased in the stronger model where the query algorithm may not access the input graph $G$, but the user promises to only give edges $(u, v) \in E(G)$ as queries. When we receive an edge query $(u, v)$, we process it as follows.
\begin{mdframed}[backgroundcolor=gray!20]
\textbf{\sss{} Adjacency Oracle Query}
\begin{itemize}
\item If $(u, v)$ was recorded by the set data structure as one of the edges added to $E(H)$ during preprocessing, then answer \yes{}.

\item Else if $u, v$ are in different connected components according to the union-find data structure, then answer \yes{}.

\item Else answer \no{}.
\end{itemize}
\end{mdframed}

\begin{lemma}
The set of edges to which the adjacency oracle answers \yes{} spans the input graph $G$.
\end{lemma}
\begin{proof}
There are two kinds of edges to which the oracle answers \yes{}.
Some are the edges added to $E(H)$ during preprocessing.
By construction, these edges span each individual connected component in the data structure at the end of preprocessing.
The query algorithm will then answer \yes{} to all additional edges in $E(G)$ between these connected components, and the lemma follows.
\end{proof}

\begin{lemma}
With high probability, the query algorithm will only answer \yes{} to at most $(1+\eps)n$ edges in total.
\end{lemma}
\begin{proof}
First, the edges added to $E(H)$ during preprocessing form a forest, so there are at most $n-1$ such edges.
Second, by Lemma \ref{lem:upedgecount}, with high probability, for each bucket $b$ there are at most $\eps n / \log n$ edges between the connected components discovered in preprocessing that have one endpoint in bucket $b$ and the other endpoint in a bucket $b' \ge b$.
Summing over the $\log n$ buckets, the number of edges between connected components is at most $\log n \cdot \frac{\eps n }{\log n} = \eps n$.
Thus, the oracle says \yes{} to at most $(n-1) + \eps n \le (1+\eps) n$ edges in total. 
\end{proof}

\paragraph{Alternate Proof Using Random $k$-Out Orientations.} After this paper was initially released, we realized that our adjacency oracle for \sss{} can also be proved as a corollary of a nice recent paper by Holm, King, Thorup, Zamir, and Zwick \cite{HKTZZ19}.
Their main result is the following structural theorem:
\begin{theorem} [\cite{HKTZZ19}]
Let $G$ be any undirected unweighted graph, and let $H$ be a random edge-subgraph obtained by selecting $k$ edges incident to each node uniformly at random and including them in $H$.
Then, if $k \ge c \log n$ for a constant $c$, then the expected number of edges between connected components of $H$ is $O(n/k)$.
\end{theorem}

It follows that we can construct an adjacency oracle that realizes Theorem \ref{thm:introsss} by sampling $k = \Oish(1)$ many edges incident to each node and memorizing a spanning forest of the resulting subgraph $H$.
There will be $O(n/k) \ll \eps n$ many edges between components, and the rest of the proof follows by the same analysis as above.
In fact, our proof above can be \emph{interpreted} as an alternate proof of a result similar to the theorem of \cite{HKTZZ19}: one can argue that with high probability, our algorithm will query only $\Oish(1)$ many edges incident to each vertex, and we have proved that $\le O(\eps n)$ edges then go between the discovered connected components.
Our proof is qualitatively different from the one in \cite{HKTZZ19}; subjectively it is a bit simpler, but it loses several log factors that are optimized out in this prior work.

\subsection{Adjacency Oracles for Sparse $k$-Connectivity Certificates}

We next extend our method to constructing adjacency oracles for sparse \textit{$k$-connectivity certificates}.
Recall that an edge subgraph $H \subseteq G$ is a $k$-(edge) connectivity certificate of $G$ if, after deleting any $(k-1)$ edges in $H$ from both $G$ and $H$, the connected components of $H$ and $G$ are the same. Then, there is a simple generalization of \sss.

\begin{mdframed}[backgroundcolor=gray!20]
\texttt{Sparse $k$-Connectivity Certificate} ($k$-\sss{}):
\begin{itemize}
\item \textbf{Input:} An $n$-node undirected graph $G = (V, E)$, and $\eps > 0$.
\item \textbf{Output:} An edge-subgraph $H$ on at most $(1+\eps)kn$ edges that forms a $k$-connectivity certificate of $G$.
\end{itemize}
\end{mdframed}

(We recall from the introduction that $1$-\sss{} and \sss{} are the same problem.)
The standard algorithm for computing a minimum-size $k$-connectivity certificate $H$ of an input graph $G$ is to repeat the following process $k$ times: compute a spanning forest of $G$, add it to $H$, and remove its edges from $G$.
We will show the analysis here, since it will provide helpful intuition for our adjacency oracle.
The following argument is from \cite{NagamochiI92}.

\begin{claim} \label{cla: k-connectivity certificate intuition}
The subgraph $H$ from the above process is a $k$-connectivity certificate of $G$.
\end{claim}
\begin{proof}
First, note that for every cut $V_1 \sqcup V_2$ of $G$, either at least $k$ edges are added crossing the cut or all the edges crossing the cut are added. This is because, for each iteration $i$ such that there is still a remaining edge crossing the cut, $T_i$ must include an edge crossing the cut since it is a spanning forest. The fact that there are $k$ iterations implies the result.

Now, let $F \subseteq E(H)$ with $|F| = k - 1$. Let $C$ be a connected component in $G \setminus F$. Suppose for the sake of contradiction that $C$ is disconnected in $H \setminus F$. Then, there exists some cut $V_1 \sqcup V_2$ of $G$ such that the vertices in $C$ are split between the sides of the cut, $G \setminus F$ has edges across the cut, and $H \setminus F$ has no edges across the cut. But, from the above, $H$ contains either all or at least $k$ edges of every cut in $G$. Hence, if $V_1 \sqcup V_2$ had at least $k$ edges crossing the cut, then $H \setminus F$ would have at least one edge crossing the cut (since $|F| = k-1$). Then, since there must be at most $k - 1$ edges crossing the cut, $G$ and $H$ have the same edges crossing the cut, yielding a contradiction.
\end{proof}


Given this, a natural algorithm to compute an adjacency oracle for $k$-\sss{} is to somehow apply the \sss{} algorithm $k$ times in a row, iteratively removing the subgraphs from $G$ in each round, and then unioning the final spanning subgraphs.
With this intuition in mind, we have the following result.

\begin{theorem}
    Let $G$ be an $n$-node graph, $\varepsilon > 0$, and $k = \Oish(1)$. Then, there exists an $\Oish(n \varepsilon^{-1})$ time algorithm that with high probability solves $k$-\sss{}, outputting the edge-subgraph $H$ as an adjacency oracle.
\end{theorem}

Rather than describing our algorithms for $k$-\sss{} from scratch, we will describe how they are obtained from our algorithms for \sss{}.
\begin{itemize}
\item We construct $k$ separate adjacency oracles, which we will label $A_1, \dots, A_k$.
We construct these sequentially.
Our query algorithm will query each edge $(u, v)$ in all of these adjacency oracles, and return the \texttt{OR} of their outputs.
(This incurs a $k$-factor in query time, which gives the requirement $k = \Oish(1)$.)

\item Each oracle $A_i$ is made using the previous algorithm for \sss{}, with two changes:
\begin{itemize}
\item Each time we sample an edge $(u, v)$, we query that edge in all previously-computed adjacency oracles $A_1, \dots, A_{i-1}$.
If any of these previous oracles answer \yes{}, then we consider this to be a failure sample and we discard it (even if it meets all the previous criteria for success from before).

\item We repeat until we see $c i \eps^{-1} 2^b \log^2 n$ consecutive failures, before setting $b \gets b+1$ and moving onto the next bucket.
(Note the additional factor of $i$ in that expression, relative to the \sss{} algorithm.)
\end{itemize}
\end{itemize}







We can show that $H$ is a $k$-connectivity certificate via a proof analogous to Claim \ref{cla: k-connectivity certificate intuition}.
\begin{lemma} \label{lem:kcccorrect2}
$H$ is a $k$-connectivity certificate of $G$.
\end{lemma}
\begin{proof}

Note that the \sss{} query algorithm always yields a spanning subgraph, irrespective of the edges recorded in preprocessing since we answer \yes{} to every edge between nodes in different connected components. Our first modification to the \sss{} algorithm amounts to considering an edge-subgraph of $G$ with all edges recorded in previous iterations deleted. 

Now, fix a cut $V_1 \sqcup V_2$ in $G$. First, observe that either $H$ contains all edges crossing the cut or at least $k$ edges crossing the cut. If in any iteration we do not record a new edge crossing the cut, then the two sides of the cut are disconnected and in separate connected components in that iteration. Hence, the adjacency oracle corresponding to that iteration will answer \yes{} to the query of each edge crossing the cut, and $H$ will contain all edges crossing the cut. Otherwise, we will record a new edge crossing the cut in all $k$ iterations, yielding $k$ total edges crossing the cut since we disregard edges which previous oracles answer \yes{} to when recording new edges.


Now, fix $F \subseteq E(H)$ with $|F| = k-1$. Let $C$ be a connected component in $G \setminus F$. Suppose, for the sake of contradiction, that $C$ is disconnected in $H \setminus F$. Then, there exists some cut $V_1 \sqcup V_2$ such that the vertices in $C$ are split across the cut, there are no edges in $H \setminus F$ crossing the cut, and $G \setminus F$ has at least one edge crossing the cut. But then, from the above, either $H$ contains all edges crossing the cut or at least $k$ edges crossing the cut. In either case, $H \setminus F$ would then have to have at least one edge crossing the cut, yielding the desired contradiction.
\end{proof}

As in the analysis of the \sss{} algorithm, we prove a lemma to control the number of edges of $G$ going between components for each the adjacency oracles.
\begin{lemma}
For all adjacency oracles $A_i$, just before we increment $b$ in the preprocessing for the $i$\textsuperscript{th} adjacency oracle, with high probability there are at most $\frac{\varepsilon n}{\log n}$ success edges remaining in $G$.
Consequently, with high probability, there are at most $(1+\eps)n$ edges in $E(G)$ for which $A_i$ answers \yes{}, but all previous adjacency oracles answer \no{}.
\end{lemma}
\begin{proof}
The proof is nearly identical to the proof of Lemma \ref{lem:upedgecount}.
The only real difference is that there is a new type of failure edge, the edges to which the previous oracles say \yes{}.
We will call these edges \textbf{deleted edges}.

Let us assume for now that the lemma holds for all previous adjacency oracles $A_1, \dots, A_{i-1}$, and we will prove it for $A_i$ (the base case is Lemma \ref{lem:upedgecount} from before).
Thus, by assumption the number of deleted edges is at most $i \cdot (1+\eps)n$.
If at most half of the failure edges are deleted edges, then the result follows from the same analysis as before.
Otherwise, if the deleted edges dominate the other types of failure edges, then if there are at least $\frac{\eps n}{\log n}$ success edges remaining, we can compute that the probability that a particular edge sample is a success is at least
$$\Theta\left( \frac{\eps n}{i \cdot (1+\eps) n \log n} \right) = \Theta\left( \frac{\eps}{i \cdot \log n} \right).$$
Thus, by standard Chernoff bounds, with high probability we will not sample $c i \eps^{-1} 2^b \log^2n$ consecutive failure edges.
The rest of the proof is completed as in Lemma \ref{lem:upedgecount}.
\end{proof}

It remains to analyze the runtime.
\begin{lemma}
The total preprocessing time for this algorithm is $\Oish(n \varepsilon^{-1})$.
\end{lemma}
\begin{proof}
For each oracle $A_i$, each bucket $b$ holds at most $n/2^b$ connected components.
We therefore have at most $n/2^b$ successes while considering bucket $b$, since every success decreases the number of components in bucket $b$ by at least one.
We sample $\Oish(\varepsilon^{-1} 2^b)$ edges between successes, and so we iterate at most $\Oish(\varepsilon^{-1} n)$ times for bucket $b$.
Since there are $\log n$ buckets per adjacency oracle, and $k = \Oish(1)$ many adjacency oracles, we get our desired preprocessing time.
\end{proof}

\section{Spanning Adjacency Oracle Lower Bounds}
In this section we show two lower bounds. First we show that no algorithm with $o(n^{1 - \varepsilon})$ preprocessing time can compute a spanning subgraph adjacency oracle on a linear number of edges. This lower bounds hold even when the adjacency oracles have access to the adjacency list and adjacency oracle of the input graph. This is the \textit{weaker} model from the introduction.

Second, we show that no randomized algorithm with $o(n)$ preprocessing time can output a spanning subgraph adjacency oracle on even $o(n^2)$ edges in general. The caveat for this stronger result is that it applies in a slightly weaker model: we assume that the adjacency oracle does not have access to the adjacency oracle or adjacency list of the input graph. Rather, we assume we are instead promised that each edge query received by the adjacency oracle is a query of an edge from $G$. This is the \textit{stronger} model from the introduction.
\subsection{Lower Bounds in the Weaker Model}
In this section we show a close to linear lower bound on the preprocessing time of any algorithm computing an $O(n)$-size spanning subgraph adjacency oracle with high probability. We appeal to a construction from \cite{PRVY19} in the setting of Local Computation Algorithms (LCAs) for spanners.

\begin{theorem} \label{thm: strong lower bound}
For all absolute constants $\delta > 0$, any algorithm that computes an $O(n)$-size spanning subgraph adjacency oracle with probability at least $2/3$ must take $\Omega(n^{1-\delta})$ preprocessing time.
\end{theorem}
\begin{proof}
Fix $\delta > 0$. From the proof of Theorem 1.3 of \cite{PRVY19}, there exists a family of random $n^{\delta/4}$-regular graphs on $n^{\delta/2}$ nodes such that any local computation algorithm outputting an $o(n^{3\delta/4})$ edge spanning subgraph with probability at least $2/3$ requires at least $\Omega(n^{\delta/4})$ queries of the graph. One important property of this  family is that the randomly generated graph is connected with high probability.

Now, let our input graph be drawn from the family of disjoint unions of $n^{1-\delta/2}$ copies of graphs independently generated from this family. A spanning subgraph of any graph from this family is the disjoint union of spanning subgraphs on each connected component. 

Consider any algorithm running in $o(n^{1-\delta})$ preprocessing time which computes a spanning subgraph adjacency oracle with probability at least $2/3$. At least half of the connected components must be unvisited during preprocessing. Hence, queries of edges from those connected components amount to running a local computation algorithm for those subgraphs with $\Oish(1)$ queries. Then, from the result of \cite{PRVY19}, since $\Oish(1) = o(n^{\delta/4})$ and the algorithm computes a spanning subgraph adjacency oracle with probability at least $2/3$, these connected components must contribute $\Omega(n^{1 + \delta/4})$ edges to the spanning subgraph induced by the adjacency oracle. This implies the desired result.

\end{proof}

A possible objection to Theorem \ref{thm: strong lower bound} is that the constructed hard instances are not connected. It turns out that it is relatively simple to get the same lower bound while also assuming connected input graphs. After generating the random connected component subgraphs $\{C_i\}$ as in Theorem \ref{thm: strong lower bound}, introduce a single auxiliary node $r$ and connect it to one node in each $C_i$ chosen at random. This leaves the asymptotic average degree unchanged. Since the $C_i$'s are themselves connected with very high probability, the resultant graph formed after adding these edges is connected with high probability. Observe that any spanning subgraph of an input graph from this family is exactly spanning subgraphs of the $C_i$'s connected by the newly introduced star of edges around $r$. Although one node in each $C_i$ is now degree $n^{\varepsilon/4} + 1$, this information does not distinguish between nodes in the induced subgraph of $C_i$ (e.g., see the proof of Theorem 1.3 in \cite{PRVY19}). Hence, the same analysis as in Theorem \ref{thm: strong lower bound} applies.

\subsection{Lower Bounds in the Stronger Model}
In this section we show that any algorithm that computes an $o(n^2)$-size spanning subgraph adjacency oracle with high probability must take $\Omega(n)$ preprocessing time. However, the caveat is that this lower bound holds under the algorithmic model where queries to the adjacency oracle only have access to information stored during preprocessing. Namely, they do not have access to the adjacency oracle and adjacency list of the underlying graph accessible previously during preprocessing. We are also promised that every query of the spanning adjacency oracle is a query of an edge that exists in the initial graph. (Without this promise, the family of random spanning trees on $n$ nodes yields the lower bound trivially.) Notably, our algorithms provide upper bounds in this model as well.

\begin{theorem} \label{thm: weak model clique lb}
Any algorithm that with constant probability computes an $o(n^2)$-size spanning subgraph adjacency oracle (without the oracle maintaining access to the input graph after preprocesssing) must take $\Omega(n)$ preprocessing time.
\end{theorem}
\begin{proof}
Suppose that we are promised that our input graph will be two disjoint $n/2$-cliques joined by one random cut edge. The input graph is sampled by uniformly selecting a partition of $[n]$ into two parts of size $n/2$ and then selecting a pair of nodes between the two parts uniformly at random. It suffices to prove the result for this restricted set of input graphs.

Now, note that every degree oracle access will return $n/2 - 1$ except for queries to the endpoints of the random cut edge. That is, degree queries only reveal whether or not a vertex is an endpoint of the cut edge.
Moreover, since the graph is two cliques joined by a cut edge, adjacency queries amount to revealing an additional node or two in a single clique. Every such query can then be viewed as revealing the clique assignment of two nodes (since in reality it reveals at most that much information). In particular, adjacency queries do not reveal any information about the location of the endpoints of the cut edge unless they actually involve one of those endpoints.

Then, for any $f(n) = o(n)$, we may assume that in any $f(n)$ preprocessing time algorithm, the preprocessing phase amounts to revealing the clique assignment of $f(n)$ random nodes in each clique (since it reveals at most that many) and checking whether each of these nodes are an endpoint of the cut edge. In particular, with probability at least $1 - \Theta(\frac{f(n)}{n})$ neither endpoint of the cut edge appears in a query.

Next, consider the subgraph induced by the unqueried nodes. With probability at least $1 - \Theta(\frac{f(n)}{n})$, it is composed of two cliques of size $\Theta(n)$ connected by a cut-edge. In particular, queries to the data structure formed by the $f(n)$ preprocessing time reveal nothing beyond this about the unqueried nodes. Assume that the cut edge is in this induced subgraph. On a query of the cut-edge, the data structure must return \yes{} with at least constant probability (with the randomness here from preprocessing). If it returns \no{}, then the graph induced by the adjacency oracle is not spanning. However, none of the queried edges can be distinguished by the preprocessing data structure and, by the randomness of the graph input, any algorithm that treats the edge queries differently is equivalent to one that treats all edge queries identically. (The protocols are averaged over the random graph inputs since the cut-edge has an equal probability of being any of the edges between nodes not queried in preprocessing.) However, there are $\Omega(n^2)$ edges with both endpoints unqueried. Hence, the graph underlying the adjacency oracle will have $\Omega(n^2)$ edges, implying the desired result.
\end{proof}

\section{Constructing Adjacency Oracles for $3$-spanners}
First, as a warm-up, we outline the construction of the adjacency oracle for approximately regular graphs.
Then we extend the construction to general graphs.

\subsection{Approximately Regular Graphs}
In this section, we will assume that all vertices in the input graph $G(V, E)$ have degree within a constant factor $C$ of $D$. We will generalize this to graphs of arbitrary degree in the next part.
The following algorithm is phrased in the stronger model where the query algorithm may not access $G$, but we are promised that only edges from $E(G)$ will be queried.
(Recall that if the query algorithm may access $G$, then this promise is not necessary, since we can first check whether the queried edge is in $E(G)$ and answer \no{} if not.)

\paragraph{Preprocessing Algorithm.}
If $D \leq n^{0.5}$, we can use the input graph itself as our adjacency oracle: that is, we can simply answer \yes{} to every query.
Otherwise, we construct the data structure for our adjacency oracle as follows:

\begin{mdframed}[backgroundcolor=gray!20]
\textbf{$3$-Spanner Adjacency Oracle Preprocessing (Approximately Regular Setting)}
\begin{itemize}
    \item Iterate over each $v \in V$ and, independently with probability $\frac{100 C \log n}{D}$, initialize a new cluster $C_v$, for which $v$ is the \textit{cluster center}.
    
    \item For all cluster centers $c$ in an arbitrary order, loop over the edges $(c,w)$ incident to $c$ and check whether $w$ is currently assigned to a cluster.
    If $w$ is not assigned to a cluster, set its cluster membership to be the cluster centered at $c$ and record the edge $(c,w)$.
    
    \item For each $v \in V$, sample $100C r \log n$ incident edges uniformly (where $1 \leq r \leq CD$ is a parameter to be specified later). For each edge $(v,w)$, if $w$ is in a different cluster than $v$ and $v$ has no edges recorded to $w$'s cluster, record that $v$ is adjacent to $w$'s cluster, and record the edge $(v,w)$. 
\end{itemize}
\end{mdframed}

We can implement our data structure with:
\begin{itemize}
    \item An array, indexed by $v \in V$, with cell corresponding to $v$ storing a set data structure supporting $\Oish(1)$ time insertions and queries (for checking whether or not $v$ is adjacent to a given cluster).
    \item Another set data structure supporting $\Oish(1)$ time insertions and queries (for checking whether a given edge is recorded). 
\end{itemize}

By the Chernoff bound, with high probability $\Oish(n/D)$ cluster centers are selected in the first step.  Then, the second step takes $\Oish(n)$ time since each vertex is of degree at most $CD$. Since checking whether a vertex $v$ is adjacent to a given cluster takes $\Oish(1)$ time using a dictionary, then the third step takes $\tO(nr)$ time. In total, this preprocessing takes $\tO(nr)$ time. An important feature of the construction is that all vertices will belong to a cluster:
\begin{claim} \label{cla: 3-spanner approx regular clustered}
With high probability, every vertex $v \in V$ is assigned to a cluster in our algorithm.
\end{claim}
\begin{proof}
Fix $v \in V$. Let $X_u$ be the indicator random variable for vertex $u \in N[v]$ (the closed neighborhood of $v$) being selected as a cluster center. Then, if $\sum X_u \geq 1$, $v$ will be clustered. Since $v$ has degree at least $D / C$, the expectation of this sum is at least
\[
\frac{D}{C} \cdot \frac{100C}{D} \cdot \log n = 100 \log n.
\]
Then, by the Chernoff bound, since the $X_u$'s are independent, the probability that the sum of the $X_u$'s is less than one is at most  
\[
1 - \exp\left(\frac{-100\log n}{4}\right) = 1 - 1/n^{25}.
\]
Then, by the union bound, all vertices are then clustered with high probability.
\end{proof}

In the following, we will assume that this high-probability event occurs, and that all nodes are clustered.

\paragraph{Query Algorithm.}

On query $(s,t)$, the adjacency oracle responds as follows (with \yes{} meaning that the underlying $3$-spanner has edge $(s,t)$ and \no{} meaning that it does not have the edge):

\begin{mdframed}[backgroundcolor=gray!20]
\textbf{$3$-Spanner Adjacency Oracle Query (Approximately Regular Setting)}
\begin{itemize}
    \item If the edge $(s,t)$ was recorded in preprocessing, output \yes.
    \item Else if $s,t$ are in different clusters, \textbf{and} we did not record an edge from $s$ to the cluster containing $t$, \textbf{and} we did not record an edge from $t$ to the cluster containing $s$, output \yes{}.
    \item  Otherwise, output \no{}.
\end{itemize}
\end{mdframed}

We can straightforwardly check which case holds in $\Oish(1)$ time.
We next observe correctness of the spanner:
\begin{claim} \label{cla: approx reg 3-spanner prop}
The subgraph $H \subseteq G$ induced by the adjacency oracle is a $3$-spanner of the input graph $G$.
\end{claim}
\begin{proof}
By standard reductions \cite{ADDJS93, ABSHJLKS20}, it suffices to verify that for each edge $(s, t) \in E(G)$ for which the oracle responds \no{} to the query $(s, t)$, there exists an $s \leadsto t$ path of length $\le 3$ of edges to which the oracle responds \yes{}.
There are two cases in which the oracles responds \no{}:
\begin{itemize}
\item First, suppose that $s, t$ are in the same cluster and that neither $s$ nor $t$ are the center $c$ of this cluster (otherwise we would record the edge $(s, t)$).
Then we record edges $(s, c)$ and $(c, t)$, so $s \to c \to t$ forms a $2$-path in $H$.

\item Second, suppose that $s$ is adjacent to the cluster of $t$ via some recorded edge $(s, x)$. Let $c$ be the center of the cluster of $t$. Then, we record each edge in the $3$-path $s \to x \to c \to t$.
(This also handles the case in which $c = t$, in which case the above path is a $2$-path of recorded edges.)
The case in which $t$ is adjacent to the cluster of $s$ via a recorded edge follows analogously.
\end{itemize}
\end{proof}

\begin{claim} \label{cla: 3-spanner approx reg edges}
With high probability, the total number of edges in the subgraph induced by the adjacency oracle is $\tO(n^2/r)$.
\end{claim}
\begin{proof}
With high probability by the Chernoff bound, there are $\Oish(n/D)$ clusters. We assume this event occurs in the following analysis. We consider the cases in which the the adjacency oracle says \yes{}:
\begin{itemize}
\item The oracle says \yes{} to all recorded edges.
The number of recorded edges between cluster centers and vertices is $O(n)$ since each vertex belongs to at most one cluster.
The number of sampled edges that we record is at most $\tO(n^2 / D)$, since there are $n$ vertices, $\tO(n/D)$ clusters, and we record at most one edge from each vertex to each cluster.
Note that, since we sample $\Oish(r)$ edges from each vertex, we may assume that $r$ is at most the max degree $O(D)$, so this bound is $\Oish(n^2 / r)$.

\item The oracle also says \yes{} to some non-recorded edges $(u, v)$, so long as $u, v$ lie in different clusters and we have not recorded an edge from $u$ to the cluster of $v$, or vice versa.
In order to bound these edges, we make two observations:
\begin{itemize}
\item The number of vertex-cluster pairs is at most $\tO(n^2 / D)$ with high probability.

\item Fix a vertex-cluster pair, $(v, X)$. If $r \geq \frac{D}{100 \log n}$, we sample all edges in $G$ the edge-sampling step, so if there are any edges between $v$ and $X$ then one will be recorded in pre-processing. Otherwise, assume $r < \frac{D}{100 \log n}$ and
suppose that $v$ has $x$ edges to $X$ in $G$.
Since $\deg(v) \le C \cdot D$, each edge incident to $v$ is sampled with probability at least $\frac{100C r \log n}{C D} = \frac{100r \log n}{D}$.
The expected number of sampled edges from $v$ to $X$ is then at least $\frac{100r x \log n}{D}$.
So, if $x \geq D / r$, then, by the Chernoff bound, one of these $x$ edges will be sampled with probability at least
\[
1 - \exp\left(\frac{-100\log n}{4}\right) = 1 - 1/n^{25}.
\]
Note that the indicator random variables for whether each edge incident to $v$ is selected  are actually not independent. Nonetheless, they are negatively correlated and the Chernoff bound still applies to negatively correlated random variables. Union-bounding over each vertex-cluster pair, at least one of the edges will be sampled for all vertex-cluster pairs $(v, X)$ with $x \ge D / r$ with probability at least $1 - 1/n^{23}$. So, with high probability, the maximum number of edges added in this case is
$$\tO\left(\frac{n^2}{D} \cdot \frac{D}{r}\right) = \tO(n^{2}/r).$$
\end{itemize}
\end{itemize}
Applying the union bound then yields the result.
\end{proof}

\begin{theorem} \label{thm: 3-spanner reg theorem}
With high probability the above construction yields an adjacency oracle for a $3$-spanner of $G$ with $\tO(n^2 / r)$ edges, in $\tO(nr)$ preprocessing time. 
\end{theorem}

We can set the parameter $r$ to be any value between $1$ and $CD$. In particular, we have the following.
\begin{corollary}
By setting $r = n^{0.5}$, with high probability the above construction yields an adjacency oracle for a $3$-spanner of $G$  with $\tO(n^{1.5})$ edges, in $\tO(n^{1.5})$ preprocessing time. 
\end{corollary}
Namely, for $D = \omega(n^{0.5})$, this algorithm computes an adjacency oracle of an optimal size $3$-spanner of $G$ in sublinear time.

\subsection{Extending to General Graphs} \label{sec: 3-spanner gen graphs}
To extend the previous algorithm to general graphs $G(V,E)$, we will create $\log n$ copies of the previous query data structure with each corresponding to a ``bucket'' of possible node degrees. On the query of a edge, we will pass the query to all of the data structures and return \yes{} if any data structure returns \yes{}, or \no{} if all data structures return \no{}. The idea is that the data structure that correctly guesses the minimum degree of the endpoints of the edge will ensure the $3$-spanner stretch property for this edge. For vertices of especially low degree, e.g., less than $n^{0.5}$, we will just always say \yes{} to the edge.
Note that these buckets are \textit{not} guessing the degree of \emph{both} endpoints of each edge $(u, v)$ being queried, but rather, the \emph{minimum} degree between $u, v$.
In this sense, we are not directly reducing to the previous analysis. Nonetheless, the high-level idea for each bucket data structure mirrors that of the approximately regular case.

Let us overview the technical reasons why we need to take this bucketing approach.
The issue is that there is a tension between sampling enough cluster centers to ensure that each node is clustered, and the number of edges added at query time (which were not recorded in preprocessing).
If we attempt the previous algorithm: for a graph with minimum degree $\delta$, we need to sample $\Oish(n/\delta)$ many cluster centers, and so there would be $\Oish(n^2/\delta)$ vertex-cluster pairs.
On the other hand, we can afford to sample $\Oish(n^{0.5})$ edges per node; for a graph with average degree $D$, this means that each edge is included with probability roughly $D/n^{0.5}$.
Following the analysis from before, we would add $\Omega(n^{3/2}D/\delta)$ total edges, which is suboptimal for graphs with $D \gg \delta$.

But, if we instead handle edges with the minimum degree endpoint of degree $\ell$ in some data structure where we sample only $\tO(n/\ell)$ cluster centers, both endpoints will be clustered with high probability. Moreover, for a given vertex-cluster pair $(v,X)$ where $v$ is degree $\ell$ and $X$ is a cluster from the data structure with $\tO(n/\ell)$ cluster centers, there can only be $\Oish(\ell/n^{0.5})$  unrecorded edges added per pair involving $v$. Otherwise, by an analysis similar to Claim \ref{cla: 3-spanner approx reg edges} we would have sampled one with high probability. Handling the edges in this way then circumvents the previous conflict to yield the desired $\tO(n^{3/2})$ edges.


Formally, we partition the interval $[n^{0.5}, n]$ into $O(\log n)$ buckets of the form $[2^kn^{0.5}, 2^{k+1}n^{0.5})$ for $k \in \{0, 1 ,\ldots, (\log n)/2 -1 \}$. For each non-empty bucket with lower bound $\ell$ there is a corresponding data structure. We describe the preprocessing and query algorithms for each individual data structure and describe how to combine these parametrized data structures to create the $3$-spanner adjacency oracle.

\paragraph{Preprocessing Algorithm.}
For each non-empty bucket with lower bound $\ell$, we construct its corresponding data structure as follows. The preprocessing is almost exactly as in the case of approximately $D$-regular graphs except using $\ell$ instead of $D$ in the number of sampled cluster centers and $C = 2$. We also include a run-time optimization in the cluster assignment step that will be especially valuable when we extend these ideas to $5$-spanners. We assume that the average degree $D$ of the graph is at least $n^{0.5}$ or else we just return the graph as the adjacency oracle. We can check whether this is the case in $O(n)$ time.

\begin{mdframed}[backgroundcolor=gray!20]
\textbf{$3$-Spanner Adjacency Oracle Preprocessing}
\begin{itemize}
    \item Iterate over each $v \in V$ and, independently with probability $\frac{c \log n}{\ell}$, initialize a distinct cluster for $v$ and assign it as its cluster's \textit{cluster center} (where $c > 0$ is a sufficiently large absolute constant we will leave implicit).
    \item Now we assign nodes to clusters.
    \begin{itemize}
        \item If $\ell \geq \sqrt{D}$, loop over each neighbor of each cluster center and assign nodes to the first cluster center they are found to be adjacent to, recording the respective edge.
        \item Otherwise, if $\ell < \sqrt{D}$, for all $v \in V$, iterate over its neighbors in a random order and assign it to the cluster of the first neighboring cluster center found, recording the corresponding edge.
    \end{itemize}
    \item For all $v \in V$, sample $cr \log n$ of its incident edges. For each $(v,w)$ sampled, if $w$ is in a different cluster than $v$, record the edge $(v,w)$ and that $v$ is adjacent to $w$'s cluster. Additional edges from $v$ to $w$'s cluster are then ignored.
\end{itemize}
\end{mdframed}

We can implement each data structure using the same data structures as in the approximately regular case (array and set data structures).

 Note that each degree is between $0$ and $n$. The sum of the degrees of the sampled cluster centers is $\frac{cnD \log n}{\ell}$ in expectation. Then, since the cluster centers are sampled independently, the average degree of the cluster centers is at most $\rho D$ with probability at least
\[
1 - \exp\left( \frac{-\rho c nD \log n}{2n\ell}\right) = 1 - n^{-c \rho D/(2\ell)}
\]
by the Chernoff bound (for independent, bounded, and non-negative random variables), where $\rho = \max(5, \ell/D)$. The Chernoff bound used here arises from the standard Chernoff bound for independent random variables in $[0,1]$ via dividing each random variable by their uniform bound (in this case $n$, since each cluster center has degree at most $n$). Then, for $c$ a sufficiently large constant, with high probability the first case of the cluster assignment step takes $\tO(\frac{n}{\ell} \cdot \rho D)$ time.  When $\max(5, \ell/D) = \ell/D$, this is $\Oish(n)$. Otherwise, this is $\tO(nD/\ell)$.

The second case of the second step takes $\tO(n \ell)$ time with high probability. Namely, for nodes of degree at least $\ell$, by the Chernoff bound, with high probability it takes checking at most $\ell$ random neighbors to find a cluster center and terminate. We can then union bound over all vertices to get the result for all vertices with high probability.

Hence, with high probability, the total preprocessing time over all data structures will be 
\[
\tO(n \min(D/\delta, \sqrt{D}) + nr),
\]
where $\delta := \max(\delta(G), n^{1/2})$. 
The variable $\delta$ corresponds to the order of the lower bound of the first bucket we instantiate. Note that the second step is somewhat different than in the approximately regular case in that we do not check all of the neighbors of $v$. This is a runtime optimization for graphs without a worst-case gap between average degree and minimum degree, but does not change the fact that all sufficiently high degree vertices will be clustered.

\paragraph{Query Algorithm.}
On the query of an edge $(s,t)$, we query all $O(\log n)$ data structures and return \yes{} if any data structure returns \yes. Otherwise we return \no. On query $(s,t)$, the data structure corresponding to lower bound $\ell$ responds as follows.
\begin{mdframed}[backgroundcolor = gray!20]
\textbf{$3$-Spanner Adjacency Oracle Query}
\begin{itemize}
    \item If $\min(\deg(s), \deg(t)) < n^{0.5}$, output \yes.
    \item Else if the edge $(s,t)$ was recorded, output \yes.
    \item Else if $s$ or $t$ is unclustered or $\min(\deg(s), \deg(t)) > 2\ell$ output \no. 
    \item Else if $s$ and $t$ are in the same cluster or $s$ or $t$ have a recorded edge to the cluster of the other, output \no.
    \item Otherwise, output \yes.
\end{itemize}
\end{mdframed}


Each query can be made in $\tO(1)$ time. (If we do not assume that the oracle has access to the input graph, we can store the degrees of each node in linear time during preprocessing.) Observe that the oracle behavior is identical to that in the case of approximately $D$-regular graphs except for the degree-checking conditions.

Intuitively, for an edge with minimum degree endpoint in the bucket $[\ell, 2\ell)$, the data structure corresponding to $\ell$ will handle the $3$-spanner property for that edge. We observe that the endpoints of the edge will be clustered in that data structure:
\begin{claim} \label{cla: gen 3-spanner clustered}
With high probability, for all bucket lower bounds $\ell$, each vertex of degree at least $\ell$ is clustered in the data structure corresponding to $\ell$.
\end{claim}
\begin{proof}
For $\ell \geq \sqrt{D}$, this follows from the Chernoff bound and the union bound via a proof nearly identical to the proof of Claim \ref{cla: 3-spanner approx regular clustered}.

For $\ell < \sqrt{D}$, let $v \in V$ be a vertex with $\deg(v) \geq \ell$. For each $u \in N[v]$, where $N[v]$ is the closed neighborhood of $v$, create an indicator random variable $X_u$ for $u$ being selected as a cluster center. The variables are independent, and their sum has expectation at least $c\log n$. Hence, by the Chernoff bound, the probability that the sum is at least $1$ is at least $1 - n^{-c/4}$, yielding the desired result after union bounding over all $v \in V$
\end{proof}

Now we can show that the subgraph $H \subseteq G$ induced by the adjacency oracle is indeed a $3$-spanner of the input graph $G$.

\begin{lemma} \label{lem: 3 span}
With high probability, the subgraph $H \subseteq G$ induced by the adjacency oracle is  a $3$-spanner of the input graph $G$.
\end{lemma}
\begin{proof}
Let $(s,t) \in E(G)$ with $\min(\deg(s), \deg(t)) \in [\ell, 2\ell)$. If $\ell < n^{0.5}$, $(s,t)\in E(H)$ since every data structure outputs \yes. Otherwise, consider the data structure corresponding to lower bound $\ell$. Both $s$ and $t$ are clustered with high probability by Claim \ref{cla: gen 3-spanner clustered}. We consider the two applicable \no{} cases and verify that the distance between $s$ and $t$ in the induced graph is at most $3$.

Again, it suffices to verify that for each edge $(s,t) \in E(G)$ for which the oracle responds \no{} to the query $(s,t)$, there exists an $s \leadsto t$ path of length $\leq 3$ of edges to which the oracle responds \yes.

\begin{itemize}
    \item Suppose that $s$ and $t$ are in the same cluster in the data structure corresponding to $\ell$ (say centered at $x$). We cannot have $s = x$ or $t = x$ because then $(s,t)$ was recorded and the data structure outputs \yes. The oracle responds \yes{} to each edge in the path $s \to x \to t$ since they are the edges recorded from $s$ and $t$ to their cluster center.
    \item Now, suppose that $s$ and $t$ belong to distinct clusters and an edge was recorded from $s$ to the cluster of $t$ (centered at $x$), with the edge being $(s,u)$ for $u$ in the cluster centered at $x$. Then, the data structure will respond \yes{} to each edge in the path $s \to u \to x \to t$. In particular, the data structure responds \yes{} to edges $(u,x)$ and $(x,t)$ since they are the edges recorded from $u$ and $t$ to their cluster centers. The data structure also responds \yes{} to $(s,u)$ since that edge was recorded in the edge sampling step. (Note that it is possible that $u = x$ or $x = t$. In either of those cases, the path is of length $2$ and is still composed of edges contained in $H$.)
\end{itemize}

\end{proof}
Finally, we bound the number of edges in the graph induced by the oracle.
\begin{lemma}
With high probability, the number of edges in the graph induced by the oracle is $\tO(n^2/\delta + n^2/r)$.
\end{lemma}
\begin{proof}
We first bound the number of edges added in each \yes{} case of the data structure corresponding to bucket with lower bound $\ell$. We will use that, with high probability, the number of clusters is $\Oish(n/\ell)$ and assume that this event holds in the below.
\begin{itemize}
    \item If $\delta$ is the minimum degree of $G$, then the first \yes{} case does not add any edges. Otherwise, it adds $O(n^{1.5})$ edges. Hence, this case adds at most $O(n^2/\delta)$ edges. 
    \item Each vertex is adjacent to at most one cluster center and has at most one edge to each cluster. Hence, the second \yes{} case adds at most $\tO(n^2/\ell) = \tO(n^2/\delta)$ edges, using the assumption of $\Oish(n/\ell)$ clusters.
    \item In the final \yes{} case, one of the endpoints is of degree at most $2 \ell$.  We can bound the number of edges added in this case by bounding the number of unsampled edges added from vertices of degree at most $2 \ell$. The number of pairs of vertices of degree at most $2 \ell$ and clusters is $\tO(n^2/\ell)$. If $r \geq \frac{2\ell}{c \log n}$, we will sample all edges for each vertex of degree at most $2 \ell$ and no edges will be added from this case.  Otherwise, assume $r < \frac{2\ell}{c \log n}$. For such a given vertex-cluster pair $(s, X)$, if $s$ has more than $\ell/r$ edges to vertices in $X$, then the expected number of sampled edges from $s$ to vertices in $X$ is at least $\frac{\ell}{r} \cdot \frac{r}{2\ell} \cdot c \log n = (c \log n)/2$. Hence, as in Claim \ref{cla: 3-spanner approx reg edges}, by the Chernoff bound (for negatively correlated random variables) and the union bound, with high probability, for each vertex cluster pair $(s, X)$ with $s$ without an edge recorded to $X$ and $s$ of degree at most $2 \ell$, $s$ has at most $\ell/r$ edges added to $X$. Hence, the total number of edges added in this case is $\tO(n^2/r)$ with high probability.

\end{itemize}
Union bounding over all $O(\log n)$ data structures, for large enough $c$, these bounds hold for all data structures with high probability. Then, the number of total edges in the underlying graph of the oracle is at most a $\log n$ factor greater than $\tO(n^2/\delta + n^2/r)$ which is still $\tO(n^2/\delta + n^2/r)$. 
\end{proof}

Union bounding all high probability bounds and choosing an appropriate $c$ (say $200$), this yields the following.
\begin{theorem}
With high probability the subgraph $H \subseteq G$ induced by the adjacency oracle is a $3$-spanner of $G$ with $\tO(n^2/\delta + n^2/r)$ edges. The oracle can be constructed in  $\tO(nr + n\min(D/\delta, \sqrt{D}))$ time, for parameter $r$ and $\delta = \max(n^{1/2}, \delta(G))$, where $\delta(G)$ is the minimum degree of a vertex in $G$.
\end{theorem}
We have the following corollary.
\begin{corollary}
By setting $r = n^{0.5}$, with high probability the above construction yields an adjacency oracle for a $3$-spanner of $G$ with $\tO(n^{1.5})$ edges, in $\tO(n^{1.5})$ preprocessing time.
\end{corollary}

\section{Adjacency Oracles for $5$-spanners}
Many facets of the construction for $3$-spanners can also apply to $5$-spanners. We apply analogous methods in this new setting. 

\subsection{Approximately Regular Graphs}
As before, we begin by assuming that all vertices in the input graph $G(V,E)$ have degree within a constant factor $C$ of $D$. Since this algorithm is almost identical to the approximately regular case for $3$-spanners, we highlight the dissimilarities.

\paragraph{Preprocessing Algorithm.}
In the case of $5$-spanners, the optimal size spanning subgraphs have $\tO(n^{4/3})$ edges. So, in this case we assume $D \geq n^{1/3}$ (otherwise outputting the graph itself as the $5$-spanner). Otherwise, we construct the data structure for our adjacency oracle via the following preprocessing steps:
\begin{mdframed}[backgroundcolor = gray!20]
\textbf{$5$-Spanner Preprocessing Algorithm (Approximately Regular Setting)}
\begin{itemize}
    \item Iterate over each $v \in V$ and, independently with probability $\frac{100 C \log n}{D}$, initialize a distinct cluster for $v$ and assign it as its cluster's cluster center.
    
    \item For all cluster centers $c$ in an arbitrary order, loop over the edges $(c,v)$ incident to $c$ and check whether $v$ is not assigned to a cluster. If $v$ is not assigned to a cluster, set its cluster membership to be the cluster centered at $c$ and record the edge $(c,v)$.
    
    \item For each $v \in V$, sample $100C r \log n$ incident edges uniformly ($r$ is an algorithm parameter as before). For each edge $(v,w)$,\textbf{ if $w$ is in a different cluster than $v$ and there are no edges yet recorded between the clusters of $v$ and $w$, record that the clusters are adjacent and record the edge $(v,w)$. }
\end{itemize}    
\end{mdframed}

The only step that differs between the $3$-spanner and $5$-spanner algorithms for approximately regular graphs is the final part of the third step. Since we are permitted to have higher stretch, we only need to maintain cluster-cluster adjacencies rather than vertex-cluster adjacencies. 

The underlying data structure created here is only slightly different to the $3$-spanner case. Again we can use arrays and set data structures. It suffices to use:
\begin{itemize}
    \item An array indexed by the vertices, storing their cluster assignments.
    \item A set data structure storing pairs of clusters with edges recorded between them.
    \item A set data structure storing recorded edges.
\end{itemize}

By the same analysis as for $3$-spanners, with high probability $\Oish(n/D)$ cluster centers are selected in the first step, preprocessing takes $\Oish(nr)$ time in total, and every vertex $v \in V$ will be clustered with high probability.

\paragraph{Query Algorithm.}

The query algorithm is also very similar. On query $(s,t)$, the oracle responds as follows. 
\begin{mdframed}[backgroundcolor = gray!20]
\textbf{$5$-Spanner Adjacency Oracle Query (Approximately Regular Setting)}
\begin{itemize}
    \item If $(s,t)$ was recorded, output \yes.
    \item Else if $s$ and $t$ are in the same cluster or \textbf{different clusters with an edge recorded between them}, output \no.
    \item Otherwise, output \yes.
\end{itemize}    
\end{mdframed}

Note that the difference in oracle behavior compared to the $3$-spanner case is that we say \no{} if $s$ and $t$ are in adjacent clusters even if neither $s$ nor $t$ are themselves adjacent to the cluster of the other.

The proof that the subgraph induced by the adjacency oracle is a $5$-spanner of $G$ differs exactly by the difference in our oracle behavior. We make the analogous, standard reduction as in Claim \ref{cla: approx reg 3-spanner prop}.

\begin{lemma} \label{lem: 5-spanner approx reg spanning}
The subgraph $H \subseteq G$ induced by the adjacency oracle is a $5$-spanner of the input graph $G$.
\end{lemma}
\begin{proof}
In the first \no{} case, when $s$ and $t$ belong to the same cluster centered at $x$, there is a path of length $2$ between $s$ and $t$ of edges for which the oracle says \yes{} (from $s$ to the shared cluster center to $t$).

 In the second case, suppose $s$ and $t$ are in clusters centered at $c_1$ and $c_2$, respectively, and the recorded edge between their clusters is $(x,y)$. Then, the path $s \to c_1 \to x \to  y \to  c_2, \to t$ is composed of edges on which the oracle says \yes. Some edges in the path may be degenerate, but, in any case, the path exists and is of length at most $5$.
 
 Hence, $H$ is a $5$-spanner of $G$.  
\end{proof}

It remains to bound the number of edges in this underlying graph. This step differs from the analysis of the $3$-spanner algorithm in the handling of the final \yes{} case of the query algorithm.
\begin{lemma}
With high probability, the subgraph $H \subseteq G$ induced by the adjacency oracle  has $\tO(n + \frac{n^2}{Dr})$ edges. 
\end{lemma}
\begin{proof}
With high probability, by the Chernoff bound, there are $\Oish(n/D)$ clusters. Then, since each vertex has a unique center and there are $\tO(n^2/D^2)$ pairs of clusters and only one edge recorded per pair, the first \yes{} case adds at most $\tO(n + n^2/D^2)$ edges.

The final \yes{} case is more complicated.  Suppose that $(C_1, C_2)$ is a pair of clusters. As before, if $r \geq \frac{D}{100 \log n}$, we iterate over all edges in the graph in the edge sampling step, so we will record an edge for each pair of clusters with edges between them, so this case will not apply. Otherwise, each edge from a vertex in $C_1$ to a vertex in $C_2$ is sampled with probability at least $\frac{100 C r \log n}{CD} = \frac{100 r \log n }{D}$. Hence, if there are more than $D/r$ edges between the two clusters, then, by the Chernoff bound for negatively correlated random variables, with  probability at least $1 - n^{-25}$, one of the edges will be sampled. Union bounding over all pairs of clusters, this holds for all pairs of clusters with probability at least $1 - n^{-23}$. Then, since in this case the endpoints of the queried edges are in non-adjacent, distinct clusters, the number of total edges added in this case is at most $\tO(n/D \cdot n/D \cdot D/r) = \tO(\frac{n^2}{Dr})$. This yields the desired result since we can assume without loss of generality that $r \leq CD$.
\end{proof}

Combining the above, we have the following.
\begin{theorem} \label{thm: regular 5-spanner result}
With high probability, the above construction yields an adjacency oracle of a $5$-spanner of $G$ with $\tO(n + \frac{n^2}{Dr})$ edges, in $\tO(rn)$ total preprocessing time. Here, the parameter $r$ satisfies $0 \leq r \leq CD$.
\end{theorem}

Setting the parameter $r$, we have the following corollary.
\begin{corollary}
By setting $r = n^{1/3}$, with high probability the above construction yields an adjacency oracle for a $5$-spanner of $G$ with $\Oish(n^{4/3})$ edges in $\Oish(n^{4/3})$ preprocessing time. 
\end{corollary}

\subsection{Extending to General Graphs}
The same bottlenecks for $3$-spanners apply to $5$-spanners: there remains a conflict between sampling enough clusters to cluster every node and adding too many unrecorded edges with the final oracle. The same bucketing trick can resolve this issue. However, a new issue arises. After selecting clusters, we assign nodes to clusters by looping over the edges incident to each cluster center. When the graph is approximately regular, this only takes $\Oish(n)$ time. However, now, if there are a linear number of low degree nodes, say of degree $\tO(n^{1/3 + \varepsilon})$, then we cannot trivially add all the edges of these nodes without losing optimal size. But, to cluster these nodes, we must sample $\tO(n^{2/3 - \varepsilon})$ many cluster centers. If the average degree of the graph is high, say $\Theta(n)$, then looping over the edges incident to the cluster centers can take $\tO(n^{5/3 - \varepsilon})$ time, much longer than the $\tO(n^{4/3})$ preprocessing time achievable in the approximately-regular case. This issue does not arise for $3$-spanners because we do not need to handle nodes of degree less that $n^{1/2}$ in that setting (so the maximum size of the gap between the minimum and average degree is sufficiently restricted).

We can somewhat overcome this bottleneck by observing that, when there are a very large number of cluster centers, we can instead sample random neighbors of nodes and check whether those neighbors are cluster centers. Regardless, for worst case inputs our preprocessing time will increase from $\tO(n^{4/3})$ to $\tO(n^{3/2})$ in extending to general graphs. This is exactly the run-time optimization included in the $3$-spanner algorithm for general graphs in Section \ref{sec: 3-spanner gen graphs}.

Formally, as for the case of $3$-spanners, we will create $O(\log n)$ copies of our data structure to handle graphs that are not approximately regular. We partition the interval $[n^{1/3}, n]$ into $O(\log n)$ buckets of the form $[2^kn^{1/3}, 2^{k+1}n^{1/3})$ for $k \in \{0, 1 ,\ldots, 2(\log n)/3 -1 \}$. As before, for each bucket with lower bound $\ell$, there is a corresponding data structure. We now describe the preprocessing and oracle behavior for these data structures. We let $D$ be the average degree of $G$; this will arise in optimizing the preprocessing time in the analysis.

\paragraph{Preprocessing Algorithm.}
For each non-empty bucket with lower bound $\ell$, we construct the corresponding data structure as follows.


\begin{mdframed}[backgroundcolor = gray!20]
\textbf{$5$-Spanner Adjacency Oracle Preprocessing}
\begin{itemize}
    \item  Iterate over each $v \in V$ and, independently with probability $\frac{c \log n}{\ell}$, initialize a distinct cluster for $v$ and assign it as its cluster's \textit{cluster center} (where $c > 0$ is a sufficiently large absolute constant we will leave implicit).
    \item Now we assign nodes to clusters.
    \begin{itemize}
        \item If $\ell \geq \sqrt{D}$, iterate over each neighbor of each cluster center and assign nodes to the first cluster center they are found to be adjacent to, recording the respective edge.
        \item Otherwise, if $\ell < \sqrt{D}$, for each node, loop over its neighbors in a random order and assign it to the first neighboring cluster center found. Upon finding a neighboring cluster center or exhausting the neighbors of the node, end the loop.
    \end{itemize}
    \item  For each $v \in V$, sample $cr \log n$ incident edges uniformly (again, $r$ is a parameter to the algorithm). For each edge $(v,w)$, if $w$ is in a different cluster than $v$ and there are no edges yet recorded between the clusters of $v$ and $w$, record that the clusters are adjacent and record the edge $(v,w)$. 
\end{itemize}
\end{mdframed}

By the exact same analysis as in Section \ref{sec: 3-spanner gen graphs}, with high probability the total preprocessing time is $\tO(n \min(D/\delta, \sqrt{D}) + nr)$, where $\delta := \max(\delta(G), n^{1/3})$. We can also use the same data structures to implement this.

\paragraph{Query Algorithm.}

Now, on query $(s,t)$, the data structure corresponding to bucket with lower bound $\ell$ responds as follows.
\begin{mdframed}[backgroundcolor = gray!20]
\textbf{$5$-Spanner Adjacency Oracle Query}
\begin{itemize}
    \item If $\min(\deg(s), \deg(t)) < n^{1/3}$, output \yes.
    \item Else if the edge $(s,t)$ was recorded, output \yes.
    \item Else if $s$ or $t$ are unclustered or $\min(\deg(s), \deg(t)) > 2\ell$ output \no.
    \item Else if $s$ and $t$ are in the same cluster or clusters with an edge recorded between them, output \no.
    \item Otherwise, output \yes.
\end{itemize}
\end{mdframed}

In general, on query $(s,t)$, we query all $O(\log n)$ data structures with edge $(s,t)$ and output \yes{} if any outputted \yes. We now analyze this adjacency oracle.

First, observe that we still have that nodes of sufficiently high degree are clustered in the data structure corresponding to lower bound $\ell$.

\begin{claim} \label{cla: clustered 5-spanner gen}
With high probability, for all bucket lower bounds $\ell$, each vertex of degree at least $\ell$ is clustered in the data structure correspond to $\ell$.
\end{claim}
\begin{proof}
This follows from the exact same proof as Claim \ref{cla: gen 3-spanner clustered} since the preprocessing is nearly identical.
\end{proof}

\begin{lemma}
With high probability, the subgraph $H \subseteq G$ induced by the adjacency oracle is a $5$-spanner of $G$.
\end{lemma}
\begin{proof}
As before, we only need to consider the cases in which the data structure corresponding to $\ell$ says \no.
Consider querying an edge $(s,t)$ where $\min(\deg(s), \deg(t)) \in [\ell, 2\ell)$ for $\ell \geq n^{1/3}$. Note that when the minimum degree of $s$ and $t$ is less than $n^{1/3}$, all data structures will output \yes{} on a query of $(s,t)$. By Claim \ref{cla: clustered 5-spanner gen},  we can ignore the first \no{} case. Then, the remaining two \no{} cases are handled by the same analysis as in Lemma \ref{lem: 5-spanner approx reg spanning}.
\end{proof}

It remains to bound the number of edges in the graph induced by the adjacency oracle.
\begin{lemma}
With high probability, the underlying graph $H \subseteq G$ induced by the adjacency oracle has $\tO(\frac{n^2}{r \delta} + \frac{n^2}{\delta^2} + n)$ edges. 
\end{lemma}
\begin{proof}
We first bound the number of edges added in each \yes{} case of the data structure corresponding to bucket with lower bound $\ell$. With high probability, there are $\Oish(n/\ell)$ clusters. We assume this event occurs throughout. The first \yes{} case adds at most $\tO(n^{4/3})$ total edges if $\delta = n^{1/3}$ and $0$ otherwise. 

Since each vertex is assigned at most one cluster center, the first part of the second \yes{} case adds at most  $O(n)$ edges. The second part of the second \yes{} case adds at most $\tO(n^2/\ell^2)$ edges since there are $\tO(n^2/\ell^2)$ pairs of clusters and we add at most one edge per pair from sampling.

As usual, the final \yes{} case is the most complicated one. Fix a pair of clusters, $(C_1, C_2)$. All edges that could be added in this case between $C_1$ and $C_2$ have one endpoint of degree at most $2\ell$. If $r \geq \frac{2\ell}{c \log n}$, we see all edges of vertices of degree at most $2 \ell$ in the edge sampling step, so we record an edge for every pair of adjacent clusters (and this case never arises). Otherwise, all such edges are sampled with probability at least $cr (\log n)/ (2\ell)$. Then, by the Chernoff bound (for negatively correlated random variables) and union bound, with high probability there are at most $\tO(\ell/r)$ edges added in this case per pair of clusters $(C_1, C_2)$ (or else, if there are more edges between $C_1$ and $C_2$ that could be added, with high probability one is added in sampling). Hence, the number of total number of edges added in this case is $\tO(\ell/r \cdot n^2/\ell^2) = \tO(\frac{n^2}{r \ell })$, with high probability.

Union bounding over all data structures, we then have that, with high probability, $|E(H)| = \tO(n + \frac{n^2}{r \delta} + \frac{n^2}{\delta^2})$.
\end{proof}

Combining all of the above  (and the probability bounds via the union bound, choosing a sufficiently large $c$ such as $c = 200$), we have the following theorem.
\begin{theorem}
With high probability, the adjacency oracle is constructed in $\tO(n\min(D/\delta, \sqrt{D}) + nr)$ time, and the subgraph $H \subseteq G$ induced by the adjacency oracle is a $5$-spanner of $G$ with $\tO(\frac{n^2}{r\delta} + \frac{n^2}{\delta^2} + n)$ edges. Here $D$ is the average degree of $G$, $\delta = \max(\delta(G), n^{1/3})$, and $r$ is a parameter.
\end{theorem}

\begin{corollary}
By setting $r = n^{1/3}$, with high probability the above construction yields an adjacency oracle for a $5$-spanner of the input graph $G$ with $\tO(n^{4/3})$ edges, in $\tO(n^{3/2})$ time.
\end{corollary}
Unlike for $3$-spanners, we lose somewhat in our preprocessing time when extending from the almost-regular setting to the general graph setting for $5$-spanners. Intuitively, this comes from the fact that assigning clusters by looping over the cluster centers incurs a cost relative to proportion of the average degree $D$ and the data structure lower bound $\ell$. This is because we must sample $\tO(n/\ell)$ cluster centers in order to cluster the low degree nodes with high probability but cluster assignment requires iterating over the adjacency lists of cluster centers of average degree $D$. For especially low $\ell$ (in particular, $\ell \leq \sqrt{D}$), we can mitigate this cost somewhat by instead looping over all vertices and assigning them the first cluster center they see. 

Nonetheless, if $D = \tO(n^{5/3})$ or $\delta = \Omega(n^{5/3})$, we achieve a sublinear preprocessing time of $\tO(n^{4/3})$, matching the time complexity achieved in the approximately regular case.

\section{Higher Stretch Spanners}
In this section, we adapt the classic algorithm of Baswana and Sen \cite{BS07} to the adjacency oracle setting to yield a sublinear time algorithm for close to optimal size $(2k-1)$-spanners for $k = \Oish(1)$. Without modification, the algorithm of Baswana and Sen takes $\Omega(m)$ time on $m$-edge graphs since it traverses the adjacency list of every vertex. However, by sampling edges from adjacency lists and deferring some computation of the $(2k-1)$-spanner to the adjacency oracle similar to how we handle cut-edges in our algorithm for \sss, we can achieve sublinear runtime. However, unlike the algorithms for $3$- and $5$-spanners, we incur an additional cost in the size of our outputted spanner from having multiple rounds of hierarchical clustering. 

As a minor aside, the algorithm of Baswana and Sen \cite{BS07} computes $(2k-1)$-spanners even for weighted graphs. Sampling edges fails dramatically for weighted graphs—having one comparatively high weight edge for each vertex forces any algorithm to scan the whole adjacency list of each node. As such, we only consider the setting of unweighted graphs.

The high-level idea of the algorithm is as follows. Each vertex is initialized to its own cluster. Then there are $k-1$ rounds of selecting clusters, re-clustering vertices, and finalizing vertices that were not re-clustered. Namely, in each round we \textit{select} each remaining cluster with probability $1/n^{1/k}$. The unselected clusters will cease to exist after the round. Each \textit{unselected} vertex $v \in V$ (vertex not in a selected cluster) then samples some fraction of the edges from its adjacency list to check for adjacency to a selected cluster. If it finds an edge to a vertex in a selected cluster, we record that edge and re-cluster the vertex in that selected cluster. Otherwise, the vertex is \textit{finalized} and, for each distinct cluster (not necessarily selected) found to be adjacent to the vertex during the edge sampling, we record one edge. In the following rounds, we ignore sampled edges to finalized nodes  and only retain selected clusters. After $k-1$ rounds are completed, we do the same finalization procedure on all the remaining non-finalized nodes.

Each cluster can be conceived of as the subtree of edges recorded when re-clustering vertices into the cluster. Note that this subtree spans the nodes in the cluster. At initialization, each cluster is diameter $0$. In each subsequent round, each selected cluster grows in diameter by at most $2$. Then, each selected cluster up to the final round has diameter at most $2k-2$. In particular, there is a path of length at most $k-1$ from each node in the cluster to the first node to be initialized to the cluster. Crucially, this means that, for each vertex $v \in V$, it suffices to record only one edge from $v$ to each adjacent cluster as that yields a path of length at most $2k-1$ of recorded edges. Within each cluster, we only need to retain the recorded edges to handle the stretch for edges between intracluster pairs of vertices. The query algorithm can then fill in the gaps at query time, adding all  queried edges not satisfying either case.

Intuitively, when a node is finalized prior to the final round, it is adjacent to few remaining, selected clusters. Otherwise, sampling edges from its adjacency list would result in the node being re-clustered. This holds for nodes that are finalized in the final round since few clusters remain at that point. The edge sampling then ensures that, for each adjacent cluster, either an edge is recorded or the vertex has few edges to that cluster (and adding all of them does not incur too much additional size).

How does this compare to our $3$- and $5$-spanner algorithms? In a sense, the $3$-spanner algorithm is an optimization of the first round of this algorithm. Instead of initializing all nodes to their own clusters, we sample nodes to choose as cluster centers and assign the remaining nodes to clusters by looping over the full adjacency lists of the sampled nodes. This optimization could apply here too but is not useful after the second round (since looping over the adjacency lists of whole clusters could then incur $\Omega(m)$ time). The $5$-spanner algorithm takes advantage of Section 4.3 of \cite{BS07}, applying a similar optimization as the one we used for $3$-spanners to a variant of the Baswana and Sen algorithm. For large stretch, it is not obvious how to extend these optimizations.

\paragraph{Preprocessing Algorithm.}
Formally, for input graph $G(V,E)$ on $n$-nodes and $m$-edges and desired stretch $2k-1$, the preprocessing algorithm is as follows. In the below, $\rho$ is a parameter that is 
 a function of $n$, controlling the fraction of the adjacency lists sampled for each vertex.

\begin{mdframed}[backgroundcolor=gray!20]
\textbf{$(2k-1)$-Spanner Adjacency Oracle Preprocessing}
\begin{itemize}
    \item Initialize the cluster assignment of each $v \in V$ to itself.
    \item Initialize $r \leftarrow 1.$ $r$ denotes the current \textit{round} of the algorithm.
    \item While $r \leq k-1$:
    \begin{itemize}
    \item Select each cluster independently with probability $1/n^{1/k}$.
    \item For each $v \in V$ that is not finalized or in a selected cluster:
    \begin{itemize}
        \item Sample $\frac{c \deg(v) \log n}{\rho}$ edges from the adjacency list of $v$ ($c > 0$ is an absolute constant that will be left implicit).
        \item For the first sampled edge $(v, w)$ with $w$ in a selected cluster (and $w$ not added to the cluster this round), record $(v,w)$ and assign $v$ to the cluster of $w$.
        \item If no such edge is found, loop over the sampled edges again and, for each edge $(v,w)$ with $w$ in a different cluster and not finalized prior to this round, record the edge $(v,w)$ if and only if an edge from $v$ to the cluster of $w$ has not yet been recorded in this step. Mark $v$ as finalized.
    \end{itemize}
    \item Increment $r \leftarrow r + 1.$
    \end{itemize}
    \item For each non-finalized $v \in V$:
    \begin{itemize}
        \item Sample $\frac{c \deg(v) \log n}{\rho}$ edges from the adjacency list of $v$. For each edge $(v,w)$ with  $w$ in a different cluster and not finalized prior to this round, record the edge $(v,w)$ if and only if an edge from $v$ to the cluster of $w$ has not yet been recorded in this step. Mark $v$ as finalized.
    \end{itemize}
\end{itemize}
\end{mdframed}

The first $k-1$ rounds of the algorithm take $\tO(n + m/\rho)$ time each. The $k$\textsuperscript{th} round also takes $\tO(n + m/\rho)$ time. We then have that preprocessing takes $\Oish(kn + km/\rho) = \Oish(n + m/\rho)$ total time, since $k = \Oish(1)$.

We can implement the underlying data structure via:
\begin{itemize}
    \item A set data structure keeping track of the recorded edges.
    \item An array indexed by vertices, storing:
    \begin{itemize}
        \item An array of the cluster assignment of the vertex in each round and whether the vertex was finalized by the beginning of that round.
        \item A set data structure storing the clusters adjacent to the vertex.
    \end{itemize}
\end{itemize}

\paragraph{Query Algorithm.}
On a query of the edge $(s,t)$, the adjacency oracle responds as follows.
\begin{mdframed}[backgroundcolor = gray!20]
\textbf{$(2k-1)$-Spanner Adjacency Oracle Query}
\begin{itemize}
    \item If $(s,t)$ was recorded during preprocessing, output \yes.
    \item Otherwise, assume that $s$ was finalized during round $r_0$ (when $r = r_0$) and $t$ was not finalized before round $r_0$ (otherwise, swap $s$ and $t$). 
    \begin{itemize}
        \item If $s$ and $t$ were in the same cluster in round at the beginning of round $r_0$, output \no.
        \item If $s$ records an edge to the cluster containing $t$ at the beginning of round $r_0$, output \no.
        \item Otherwise, output \yes.
    \end{itemize}
\end{itemize}
\end{mdframed}

Note that the query algorithm can be implemented in $\Oish(1)$ time (independent of $k$) using the data structure described above. Additionally, if $s$ and $t$ were ever in the same cluster and the edge $(s,t)$ was not recorded, we could output \no{}. However, checking this would incur a factor of $k$ in the query time, and we do not make use of this in the analysis.

We begin our analysis of the algorithm with a helpful structural claim.
\begin{claim} \label{cla: 2k-1 cluster diameter}
At the beginning of round $r$, the induced subgraph of recorded edges of each cluster has diameter at most $2r - 2$.
\end{claim}
\begin{proof}
We prove the result by induction. Label each cluster by the first vertex initialized to that cluster. Call that vertex the \textit{cluster center}. At the beginning of round $1$, no vertices are finalized and all clusters are a single vertex.

Now, suppose that, at the beginning of the $\ell$\textsuperscript{th} round, each vertex in each cluster (of non-finalized vertices) has a path of recorded edges of length at most $\ell - 1$ to its cluster center. Then, each unselected cluster ceases to exist in the following round: all of its nodes are either finalized or re-clustered. For each selected cluster, any node that is re-clustered has an accompanying recorded edge to a node that was in the cluster before this round. Then, by the inductive hypothesis, at the beginning of the $(\ell + 1)$\textsuperscript{th} round, the re-clustered node has a path of length of recorded edges of length at most $\ell$ to its cluster center. The result then follows by induction.
\end{proof}
We can now prove correctness of the query algorithm.
\begin{lemma}
The subgraph $H \subseteq G$ induced by the adjacency oracle is a $(2k-1)$-spanner of the input graph $G$.
\end{lemma}
\begin{proof}
As in Claim \ref{cla: approx reg 3-spanner prop}, it suffices to verify that for each edge $(s,t) \in E(G)$ for which the oracle responds \no{} to the query $(s,t)$, there exists an $s \leadsto t$ path of length $\leq 2k-1$ of edges to which the oracle responds \yes. There are two cases to consider.
\begin{itemize}
    \item In the first \no{} case, $s$ and $t$ were in the same cluster in some round. By Claim \ref{cla: 2k-1 cluster diameter}, since there are $k$ rounds of the algorithm (including the final partial round), there exists a path of recorded edges of length at most $2k-2$ between $s$ and $t$.
    \item In the second \no{} case, $s$ records a edge to a node $w$ in the same cluster as $t$ in some round. Then, by Claim \ref{cla: 2k-1 cluster diameter}, since there exists a path of recorded edges of length at most $2k-2$ between $w$ and $t$, there exists a path of recorded edges of length at most $2k-1$ between $s$ and $t$.
\end{itemize}
\end{proof}

\begin{lemma} \label{lem: high stretch edge bd}
With high probability, the subgraph $H \subseteq G$ induced by the adjacency oracle satisfies $|E(H)| = \Oish(n^{1 + 1/k}\rho^2)$.
\end{lemma}
\begin{proof}
We consider each \yes{} case individually. We can bound the number of edges added in preprocessing by handling them in two cases. Each vertex is re-clustered at most $k-1$ times. Hence, at most $O(nk) = \Oish(n)$ edges are recorded from re-clustering vertices. 

The remaining recorded edges are those recorded upon finalizing a vertex. Fix a vertex $v$ and suppose it is finalized in round $\ell < k$. When $v$ is finalized, with high probability it is adjacent to at most $\rho$ selected clusters. This is because, if $v$ is adjacent to at least $\rho$ selected clusters, then, by the Chernoff bound (for negatively correlated random variables), the probability that an edge to a node in a selected cluster is sampled is at least $1 - n^{-c/4}.$

Conditioned on the above event, with probability at least  $1 - n^{-c/4}$ we have that $v$ is adjacent to at most $c \rho n^{1/k} \log n = \Oish(\rho n^{1/k})$ clusters in round $\ell$. This follows from the Chernoff bound, since each cluster is selected independently with probability $1/n^{1/k}$ in each round.

Now suppose that $v$ is finalized in the $k$\textsuperscript{th} round. Note that the cluster each vertex is initialized into remains in the $k$\textsuperscript{th} round with probability $1/n^{(k-1)/k}$. All clusters are chosen from these initial clusters, so, by the Chernoff bound, there are at most $c n^{1/k} \log n$ clusters remaining in the $k$\textsuperscript{th} round with probability at least $1 - n^{-c/4}$. Hence, in this case $v$ is adjacent to $\Oish(n^{1/k})$ clusters with probability at least $1 - n^{-c/4}$ in this case.

Finally, in both cases, for each cluster adjacent to $v$ via at least $\rho$ edges, one is recorded with probability at least $1 - n^{-c/4}$, by the Chernoff bound (for negatively correlated random variables).

Union bounding each of these events (and the final event over each adjacent cluster) and over every $v \in V$, we have that, with probability at least $1 - n^{(16 -c)/4}$, every $v \in V$ is adjacent to at most $\Oish(\rho n^{1/k})$ clusters when finalized and has an edge recorded to each cluster adjacent via at least $\rho$ edges. Hence, the total number of edges recorded in preprocessing is $\Oish(n^{1 + 1/k}\rho)$ with probability at least $1 - n^{(16 -c)/4}$.

Now, we consider the final \yes{} case. Suppose that the edge $(s,t)$ was queried, $s$ was finalized in round $r_0$, and $t$ is not finalized in a round earlier than $s$. At the beginning of round $r_0$ we know that $s$ and $t$ were not in the same cluster. So, the edge from $s$ to $t$ is an edge from $s$ to a different cluster. Moreover, no edge from $s$ to the cluster of $t$ was recorded. Hence, by the above, with probability at least $1 - n^{(16 -c)/4}$ at most $\Oish(n^{1 + 1/k}\rho^2)$ edges can be added in this case. Setting $c$ sufficiently high then yields the desired result.

\end{proof}

Combining the above, we have the following.
\begin{theorem}
Let $G$ be an $n$-node, $m$-edge graph, and let $k = \Oish(1)$. Then, there exists an $\Oish(n + m/\rho)$ time algorithm that with high probability computes a $(2k-1)$-spanner $H \subseteq G$ of size $\Oish(n^{1 + 1/k}\rho^2)$, outputting $H$ as an adjacency oracle.
\end{theorem}

\section*{Acknowledgments}

We are grateful to Merav Parter, Nathan Wallheimer, Amir Abboud, Ron Safier, and Oded Goldreich for references and technical discussions on the paper.
We are also grateful to an anonymous reviewer for exceptionally helpful and thorough comments.

\bibliographystyle{plain}
\bibliography{references}

\begin{thebibliography}{10}

\bibitem{ABSHJLKS20}
Reyan Ahmed, Greg Bodwin, Faryad~Darabi Sahneh, Keaton Hamm, Mohammad
  Javad~Latifi Jebelli, Stephen Kobourov, and Richard Spence.
\newblock Graph spanners: A tutorial review.
\newblock {\em Computer Science Review}, 37:100253, 2020.

\bibitem{ARVX12}
Noga Alon, Ronitt Rubinfeld, Shai Vardi, and Ning Xie.
\newblock Space-efficient local computation algorithms.
\newblock In {\em Proceedings of the twenty-third annual ACM-SIAM symposium on
  Discrete Algorithms}, pages 1132--1139. SIAM, 2012.

\bibitem{ADDJS93}
Ingo Alth{\"o}fer, Gautam Das, David Dobkin, Deborah Joseph, and Jos{\'e}
  Soares.
\newblock On sparse spanners of weighted graphs.
\newblock {\em Discrete \& Computational Geometry}, 9(1):81--100, 1993.

\bibitem{ACLP23}
Rubi Arviv, Lily Chung, Reut Levi, and Edward Pyne.
\newblock Improved lcas for constructing spanners.
\newblock {\em arXiv preprint arXiv:2105.04847}, 2023.

\bibitem{BS07}
Surender Baswana and Sandeep Sen.
\newblock A simple and linear time randomized algorithm for computing sparse
  spanners in weighted graphs.
\newblock {\em Random Structures \& Algorithms}, 30(4):532--563, 2007.

\bibitem{girth}
Paul Erd{\H{o}}s.
\newblock Extremal problems in graph theory.
\newblock In {\em Proceedings of the Symposium on Theory of Graphs and its
  Applications}, page 2936, 1963.

\bibitem{EMR18}
Guy Even, Moti Medina, and Dana Ron.
\newblock Best of two local models: Centralized local and distributed local
  algorithms.
\newblock {\em Information and Computation}, 262:69--89, 2018.

\bibitem{HMM93}
Torben Hagerup, Kurt Mehlhorn, and James~Ian Munro.
\newblock Optimal algorithms for generating discrete random variables with
  changing distributions.
\newblock {\em Lecture Notes in Computer Science}, 700:253--264, 1993.

\bibitem{HKNO09}
Avinatan Hassidim, Jonathan~A Kelner, Huy~N Nguyen, and Krzysztof Onak.
\newblock Local graph partitions for approximation and testing.
\newblock In {\em 2009 50th Annual IEEE Symposium on Foundations of Computer
  Science}, pages 22--31. IEEE, 2009.

\bibitem{HMV16}
Avinatan Hassidim, Yishay Mansour, and Shai Vardi.
\newblock Local computation mechanism design.
\newblock {\em ACM Transactions on Economics and Computation (TEAC)},
  4(4):1--24, 2016.

\bibitem{HKTZZ19}
Jacob Holm, Valerie King, Mikkel Thorup, Or~Zamir, and Uri Zwick.
\newblock Random k-out subgraph leaves only o (n/k) inter-component edges.
\newblock In {\em 2019 IEEE 60th Annual Symposium on Foundations of Computer
  Science (FOCS)}, pages 896--909. IEEE, 2019.

\bibitem{LL18}
Christoph Lenzen and Reut Levi.
\newblock A centralized local algorithm for the sparse spanning graph problem.
\newblock In {\em 45th International Colloquium on Automata, Languages, and
  Programming (ICALP 2018)}. Schloss Dagstuhl-Leibniz-Zentrum fuer Informatik,
  2018.

\bibitem{LM17}
Reut Levi and Moti Medina.
\newblock A (centralized) local guide.
\newblock {\em Bulletin of the {EATCS}}, 122:60--92, 2017.

\bibitem{LMRRS17}
Reut Levi, Guy Moshkovitz, Dana Ron, Ronitt Rubinfeld, and Asaf Shapira.
\newblock Constructing near spanning trees with few local inspections.
\newblock {\em Random Structures \& Algorithms}, 50(2):183--200, 2017.

\bibitem{LRR16}
Reut Levi, Dana Ron, and Ronitt Rubinfeld.
\newblock A local algorithm for constructing spanners in minor-free graphs.
\newblock In {\em Approximation, Randomization, and Combinatorial Optimization.
  Algorithms and Techniques (APPROX/RANDOM 2016)}, volume~60, pages
  38:1--38:15, 2016.

\bibitem{LRR20}
Reut Levi, Dana Ron, and Ronitt Rubinfeld.
\newblock Local algorithms for sparse spanning graphs.
\newblock {\em Algorithmica}, 82(4):747--786, 2020.

\bibitem{LRY15}
Reut Levi, Ronitt Rubinfeld, and Anak Yodpinyanee.
\newblock Local computation algorithms for graphs of non-constant degrees.
\newblock In {\em Proceedings of the 27th ACM symposium on Parallelism in
  Algorithms and Architectures}, pages 59--61, 2015.

\bibitem{MPV18}
Yishay Mansour, Boaz Patt-Shamir, and Shai Vardi.
\newblock Constant-time local computation algorithms.
\newblock {\em Theory of Computing Systems}, 62:249--267, 2018.

\bibitem{MRVX12}
Yishay Mansour, Aviad Rubinstein, Shai Vardi, and Ning Xie.
\newblock Converting online algorithms to local computation algorithms.
\newblock In {\em Automata, Languages, and Programming: 39th International
  Colloquium, ICALP 2012, Warwick, UK, July 9-13, 2012, Proceedings, Part I
  39}, pages 653--664. Springer, 2012.

\bibitem{MV13}
Yishay Mansour and Shai Vardi.
\newblock A local computation approximation scheme to maximum matching.
\newblock In {\em Approximation, Randomization, and Combinatorial Optimization.
  Algorithms and Techniques: 16th International Workshop, APPROX 2013, and 17th
  International Workshop, RANDOM 2013, Berkeley, CA, USA, August 21-23, 2013.
  Proceedings}, pages 260--273. Springer, 2013.

\bibitem{NagamochiI92}
Hiroshi Nagamochi and Toshihide Ibaraki.
\newblock A linear-time algorithm for finding a sparse $k$-connected spanning
  subgraph of a $k$-connected graph.
\newblock {\em Algorithmica}, 7(5{\&}6):583--596, 1992.

\bibitem{NO08}
Huy~N Nguyen and Krzysztof Onak.
\newblock Constant-time approximation algorithms via local improvements.
\newblock In {\em 2008 49th Annual IEEE Symposium on Foundations of Computer
  Science}, pages 327--336. IEEE, 2008.

\bibitem{PRVY19}
Merav Parter, Ronitt Rubinfeld, Ali Vakilian, and Anak Yodpinyanee.
\newblock {Local Computation Algorithms for Spanners}.
\newblock In {\em 10th Innovations in Theoretical Computer Science Conference
  (ITCS 2019)}, volume 124, pages 58:1--58:21, 2018.

\bibitem{PU89sicomp}
David Peleg and Jeffrey Ullman.
\newblock An optimal synchronizer for the hypercube.
\newblock {\em SIAM Journal on Computing (SICOMP)}, 18(4):740–--747, 1989.

\bibitem{PU89jacm}
David Peleg and Eli Upfal.
\newblock A trade-off between space and efficiency for routing tables.
\newblock {\em Journal of the ACM (JACM)}, 36(3):510--530, 1989.

\bibitem{RV16}
Omer Reingold and Shai Vardi.
\newblock New techniques and tighter bounds for local computation algorithms.
\newblock {\em Journal of Computer and System Sciences}, 82(7):1180--1200,
  2016.

\bibitem{Rubinfeld17}
Ronitt Rubinfeld.
\newblock Can we locally compute sparse connected subgraphs?
\newblock In {\em Computer Science--Theory and Applications: 12th International
  Computer Science Symposium in Russia, CSR 2017, Kazan, Russia, June 8-12,
  2017, Proceedings 12}, pages 38--47. Springer, 2017.

\bibitem{RTVX11}
Ronitt Rubinfeld, Gil Tamir, Shai Vardi, and Ning Xie.
\newblock Fast local computation algorithms.
\newblock {\em arXiv preprint arXiv:1104.1377}, 2011.

\bibitem{Solomon21}
Shay Solomon.
\newblock Local algorithms for bounded degree sparsifiers in sparse graphs.
\newblock {\em arXiv preprint arXiv:2105.02084}, 2021.

\bibitem{VY21}
Nithin Varma and Yuichi Yoshida.
\newblock Average sensitivity of graph algorithms.
\newblock In {\em Proceedings of the 2021 ACM-SIAM Symposium on Discrete
  Algorithms (SODA)}, pages 684--703. SIAM, 2021.

\bibitem{YYI12}
Yuichi Yoshida, Masaki Yamamoto, and Hiro Ito.
\newblock Improved constant-time approximation algorithms for maximum matchings
  and other optimization problems.
\newblock {\em SIAM Journal on Computing}, 41(4):1074--1093, 2012.

\end{thebibliography}
\end{document}